\documentclass[conference]{IEEEtran}
\IEEEoverridecommandlockouts
\pdfoutput=1

\usepackage{setspace}
\usepackage{color}
\usepackage{ulem}
\usepackage{bigstrut}
\usepackage{diagbox}
\usepackage{microtype}
\usepackage[misc]{ifsym}
\usepackage{cite}
\usepackage{amsmath,amssymb,amsfonts}
\usepackage[linesnumbered,ruled,vlined]{algorithm2e}

\SetCommentSty{mycommfont}
\usepackage{footmisc}
\makeatletter
\renewcommand{\maketag@@@}[1]{\hbox{\m@th\normalsize\normalfont#1}}%
\makeatother
\SetKwInput{KwInput}{Input}                
\SetKwInput{KwOutput}{Output}              
\usepackage{enumitem}
\usepackage{graphicx}
\usepackage{float}
\usepackage{textcomp}
\usepackage{xcolor}
\usepackage{multirow}
\usepackage{amsthm}
\usepackage{url}
\usepackage{threeparttable}
\usepackage{makecell}
\usepackage{subfigure}
\usepackage{amsmath}
\usepackage{bm}
\newcommand{\cmmnt}[1]{} 
\usepackage{lettrine}
\usepackage{booktabs} 
\usepackage[noend]{algpseudocode}
\usepackage{mathrsfs}

\usepackage{colortbl}
\usepackage{threeparttable}

\usepackage{tikz}
\newcommand*\emptycirc[1][1ex]{\tikz\draw (0,0) circle (#1);} 
\newcommand*\halfcirc[1][1ex]{%
	\begin{tikzpicture}
		\draw[fill] (0,0)-- (90:#1) arc (90:270:#1) -- cycle ;
		\draw (0,0) circle (#1);
\end{tikzpicture}}
\newcommand*\fullcirc[1][1ex]{\tikz\fill (0,0) circle (#1);} 

\def\BibTeX{{\rm B\kern-.05em{\sc i\kern-.025em b}\kern-.08em
    T\kern-.1667em\lower.7ex\hbox{E}\kern-.125emX}}
\usepackage{algpseudocode}

\newtheorem{lemma}{Lemma}

\newtheorem{Def}{Definition}

\usepackage{algpseudocode}
\usepackage{tikz,xcolor,hyperref}
\definecolor{lime}{HTML}{A6CE39}

\hypersetup{hidelinks} 

\begin{document}
\title{Optimal Hub Placement and Deadlock-Free Routing for Payment Channel Network Scalability}

\author{\IEEEauthorblockN{Lingxiao Yang$^\ast$, Xuewen Dong$^{\ast,}$\textsuperscript{\Letter}\thanks{\textsuperscript{\Letter} Xuewen Dong is the corresponding author. Permission to make digital or hard copies of all or part of this work for personal or classroom use is granted without fee provided that copies are not made or distributed for profit or commercial advantage and that copies bear this notice and the full citation on the first page. Copyrights for components of this work owned by others than the author(s) must be honored. To copy otherwise, or republish, to post on servers or to redistribute to lists, requires prior specific permission. Request permissions from pubs-permissions@ieee.org. \textit{ICDCS '23, July 18-21, 2023, Hong Kong, China.} \copyright \  2023 Copyright held by the owner/author(s). Publication rights licensed to IEEE.}, Sheng Gao$^\dagger$, Qiang Qu$^\ddagger$, Xiaodong Zhang$^\ast$, Wensheng Tian$^\S$, Yulong Shen$^\ast$}
	\IEEEauthorblockA{$^\ast$School of Computer Science and Technology, Xidian University, Xi'an, China, \\
	lxyang@stu.xidian.edu.cn, \{xwdong, zhangxiaodong\}@xidian.edu.cn, ylshen@mail.xidian.edu.cn \\
	$^\dagger$School of Information, Central University of Finance and Economics, Beijing, China. sgao@cufe.edu.cn \\
	$^\ddagger$Shenzhen Institute of Advanced Techbology, Chinese Academy of Sciences;\\ Huawei Blockchain Lab, Huawei Cloud Tech Co., Ltd, Shenzhen, China. quqiang4@huawei.com \\
	$^\S$Nanhu Lab, Jiaxing, China. tws@nanhulab.ac.cn
	}
}

\maketitle

\begin{abstract}
As a promising implementation model of payment channel network (PCN), payment channel hub (PCH) could achieve high throughput by providing stable off-chain transactions through powerful hubs. However, existing PCH schemes assume hubs are preplaced in advance, not considering payment requests' distribution and may affect network scalability, especially network load balancing. In addition, current source routing protocols with PCH allow each sender to make routing decision on his/her own request, which may have a bad effect on performance scalability (e.g., deadlock) for not considering other senders' requests. This paper proposes a novel multi-PCHs solution with high scalability. First, we are the first to study the PCH placement problem and propose optimal/approximation solutions with load balancing for small-scale and large-scale scenarios, by trading off communication costs among participants and turning the original NP-hard problem into a mixed-integer linear programming (MILP) problem solving by supermodular techniques. Then, on global network states and local directly connected clients' requests, a routing protocol is designed for each PCH with a dynamic adjustment strategy on request processing rates, enabling high-performance deadlock-free routing. Extensive experiments show that our work can effectively balance the network load, and improve the performance on throughput by 29.3\% on average compared with state-of-the-arts.
\end{abstract}


\vspace{-0.2cm}
\section{Introduction}\label{sec:introduction}
\vspace{-0.1cm}
Cryptocurrencies are gaining popularity in the financial ecosystem. However, the scalability issues of their underlying blockchain technology are still challenging. Since each transaction needs to be confirmed by the consensus mechanism, this can take several minutes to hours. Instead of continually improving the design of the consensus mechanism, a leading layer-2 proposal for addressing the scalability challenge relies on off-chain payment channels\cite{LN, RD}. The core idea is to move mass transactions submitted on-chain to off-chain and execute them securely using a locking mechanism. Only the key steps (e.g., resolving disputes, opening/closing channels) are put on-chain for confirmation.

Multiple payment channels among nodes constitute a \textit{payment channel network} (PCN). Two nodes without the direct payment channel can conduct off-chain transactions through intermediaries routing. While appealing, PCNs raise the issue of finding paths by the senders and maintaining the network topology. Furthermore, the collateral deposited in a channel cannot be used anywhere else in a bounded time. The above reasons prompted the design of TumblerBit \cite{Ethan2017TumbleBit}, which first proposes the \textit{payment channel hubs}.

\textit{Payment channel hubs (PCHs)} \textit{are the untrusted intermediaries that allow participants to make fast, anonymous, off-chain payments} \cite{Ethan2017TumbleBit}. The basic idea is that each participant opens a channel with a PCH. The PCH mediates payments between senders and recipients and gains a fee. Although this scheme sacrifices the fully decentralized nature of blockchain, it significantly improves the performance on the premise of verifiable security \cite{Tairi}. Blockchain has been compromising decentralization in the increasingly urgent need for high availability. For example, EOS \cite{EOS} has compromised from decentralization to multi-centralization.

\begin{table*}[t]
	\vspace{-10mm}
	\begin{center}
		\caption{State-of-the-art PCN scalable schemes}	\vspace{-0.3cm}
		\resizebox{2\columnwidth}{!}{
			\begin{threeparttable}
				\begin{tabular}{l | c c c c c c | c}
					\hline
					\diagbox{Properties}{Literatures} &\makecell[c]{Lighting \cite{LN}, Raidon \cite{RD}\\ Flare \cite{Flare2016}, Sprites (FC '19) \cite{0001BBKM19}}
					&REVIVE (CCS '17) \cite{KhalilG17} &Spider (NSDI '20) \cite{Sivaraman2020HighTC} &Flash (CoNEXT '19) \cite{WangXJW19} &\makecell[c]{TumbleBit (NDSS '17) \cite{Ethan2017TumbleBit}\\ A$^2$L (S\&P '21) \cite{Tairi}} &\makecell[c]{Perun (S\&P '19) \cite{Dziembowski2019Perun}\\ Commit-Chains \cite{Rami2018CommitChains}} &\makecell[c]{Splicer \\ (This work)}\\
					\hline
					Improving throughput &\fullcirc &\fullcirc &\fullcirc &\fullcirc &\fullcirc &\fullcirc &\fullcirc\\
					
					Support large transactions &\emptycirc &\emptycirc &\fullcirc &\fullcirc &\emptycirc &\emptycirc &\fullcirc\\
					
					Payment channel balance &\emptycirc &\fullcirc &\fullcirc &\emptycirc &\emptycirc &\emptycirc &\fullcirc\\
					
					Deadlock-free routing &\emptycirc &\halfcirc &\halfcirc &\emptycirc &\emptycirc &\emptycirc &\fullcirc\\
					
					Transaction unlinkability &\emptycirc &\emptycirc &\emptycirc &\emptycirc &\fullcirc &\emptycirc &\fullcirc\\
					
					Optimal hub placement &\emptycirc &\emptycirc &\emptycirc &\emptycirc &\emptycirc &\emptycirc &\fullcirc\\
					\hline
				\end{tabular}	
				\label{table_dif}
			\end{threeparttable}
		}
	\end{center}
	\vspace{-0.8cm}
\end{table*}

\textbf{Motivation.} With the increase in usage frequency, the overall load of the PCN rises rapidly, and load imbalance phenomena happen gradually. When the blockchain community\footnote{Like all public blockchains, our solution is also community-autonomous. This means that all members have equal rights in decision-making, which requires a 67\% majority approval.} decides to upgrade or design a new PCN for large-scale usage scenarios, the overall load of the PCN system should be considered. The inappropriate placement of PCHs in PCNs can easily lead to an unbalanced network communication load. However, as shown in Table \ref{table_dif}, the existing scalable schemes (e.g., \cite{LN, RD,Flare2016, 0001BBKM19, KhalilG17, Sivaraman2020HighTC, WangXJW19}) mainly study the improvement of routing strategies to improve performance. The placement problem has never been discussed, which is the \textit{first flaw}. The current source routing they adopted requires each sender to maintain a complete PCN topology and independently compute routes according to his/her own requests \cite{Sivaraman2020HighTC, WangXJW19}. In large-scale networks, the senders' performance is severely challenged, and without coordinating other requests can easily cause deadlocks. Existing PCHs (e.g., \cite{Ethan2017TumbleBit, Tairi, Dziembowski2019Perun, Rami2018CommitChains}) inherit the above source routing flaw and design complex cryptographic primitives to provide unlinkability \cite{Ethan2017TumbleBit, Tairi} for off-chain payments to obfuscate the relationship between transaction parties, with limited scalability improvements \cite{Dziembowski2019Perun, Rami2018CommitChains}, which is the \textit{second flaw}.

This paper presents \textit{Splicer}, a novel multi-PCHs solution with high scalability while inheriting unlinkability. In particular,  we propose the PCHs called the \textit{smooth nodes} for routing computations. The scalability of Splicer is two-fold: (i) Our network scale is scalable, allowing more clients to access the system through a variable number of PCHs. Meanwhile, we study the PCH placement problem to achieve optimal network communication load balancing. (ii) Splicer provides performance scalability. We design a rate-based routing protocol to perform multi-path payments by splitting transactions, achieving a high \textit{transaction success ratio} and \textit{throughput}. It proves that the network funds flow smoothly and the network is nearly deadlock-free. The \textbf{\textit{insight}} of Splicer is to seek a new tradeoff between decentralization and scalability for PCNs\footnote{Our system name Splicer is derived from this insight, which rationally connects the clients in PCNs with multiple PCHs.}.

\textbf{Challenges.} Splicer addresses several challenges: (i) \textit{PCH placement modeling and solving.} Since the clients in PCNs are scattered geographically, the desired properties are challenging to define formally. In addition, the method of solving the model may change as the scale of PCNs increases. (ii) \textit{A scalable routing protocol for PCHs.} Naive approaches (e.g., shortest-path routing) may lead to underutilizing funds or deadlock in some channels because the transaction always flows on the shortest paths. Additionally, previous instant and atomic routing \cite{LN, RD} would cause the payment value to be limited by channel funds.

\textbf{Contributions.} We make the following contributions:
\begin{itemize}[leftmargin=*]
\item[$\bullet$] \textbf{Globally optimal PCH placement.} Multiple PCHs make distributed routing computations for client payment requests to balance the network load. We model the PCH placement problem to minimize the communication delay and overhead. Besides, we propose two solutions for the placement optimization problem in small-scale and large-scale networks.

\item[$\bullet$] \textbf{Deadlock-free routing protocol.} We design a scalable rate-based routing mechanism for the multiple PCHs, allowing large-value transactions to be completed on low-capacity channels in multi-path and enabling high-throughput without disrupting the balance of channel funds. Additionally, we consider congestion control to adjust the transaction flow rates in the PCNs further.

\item[$\bullet$] \textbf{Evaluation.} We implement Splicer using Lightning Network Daemon (LND) testnet\cite{lnd}. We simulate the PCHs placement model and conduct extensive experiments to evaluate the performance of Splicer. The results show that Splicer can effectively balance the network load, and improve the performance on transaction success ratio by 42\% and throughput by 29.3\% on average compared with the state-of-the-arts.
\end{itemize}

\section{Preliminaries} \label{pre}
\subsection{Payment Channel Networks}
Fig. \ref{fig_PCN1} shows an example of payment channel networks in which $(A, C)$ and $(C, B)$ have established a bidirectional payment channel, respectively. Each direction of the bidirectional payment channels has deposited ten tokens. So a virtual payment channel \cite{Dziembowski2019Perun} is formed between $A$ and $B$. $C$ relays the payment if $A$ wants to send five tokens to $B$. Thus two transactions occur, $A$ to $C$ and $C$ to $B$. As an incentive to participate, $C$ receives a forwarding fee. A key challenge is ensuring that $C$ forwards the correct amount of tokens. Thus, the cryptographic \textit{hash time lock contract} (HTLC)\cite{LN} is proposed to guarantee that $C$ can get the tokens paid by $A$ on the channel $(A, C)$ only after $C$ has successfully paid to $B$ within a bounded time in the channel $(C, B)$.

\begin{figure}[h]
	\vspace{-0.6cm}
	\subfigure[A simple PCN.]{
		\begin{minipage}[t]{0.3\linewidth}
			\centering
			\includegraphics[width=1in]{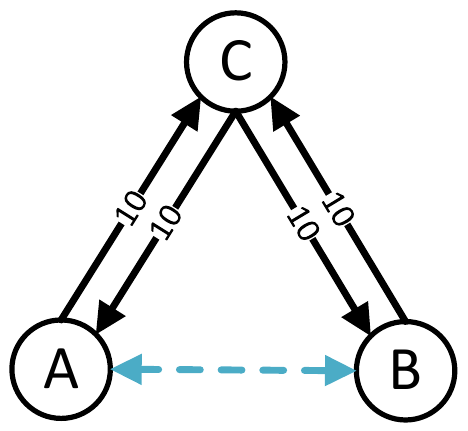}
			\label{fig_PCN1}
			\vspace{-1cm}
		\end{minipage}%
	}%
	\subfigure[The initial state.]{
		\begin{minipage}[t]{0.3\linewidth}
			\centering
			\includegraphics[width=1in]{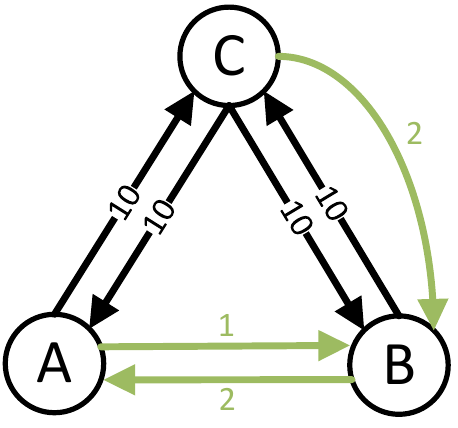}
			\label{fig_PCN2}
			\vspace{-1cm}
		\end{minipage}%
	}%
	\subfigure[A deadlock at \textit{C}.]{
		\begin{minipage}[t]{0.3\linewidth}
			\centering
			\includegraphics[width=1in]{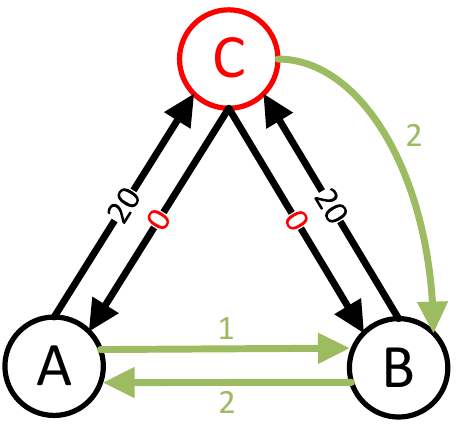}
			\label{fig_PCN3}
			\vspace{-1cm}
		\end{minipage}
	}%
	\centering
	\vspace{-0.3cm}
	\caption{Examples illustrating the PCNs.}
	\vspace{-0.6cm}
\end{figure}

\subsection{Local Deadlocks in PCNs}\label{deadlock}
To illustrate the local deadlocks in PCNs, we consider an initial state in Fig. \ref{fig_PCN2} that $A$ and $C$ transfer funds to $B$ at a rate of 1 and 2 tokens/sec, respectively. $B$ transfers funds to $A$ at a rate of 2 tokens/sec. It is noted that the payment rates we specify are not balanced, which causes net funds to flow out of $C$ and into $A$ and $B$. Payment channels are balanced by ten tokens in each direction. We suppose that the transactions only flow between $A$ and $B$. The system can achieve a total throughput of 2 tokens/sec only by allowing $A$ and $B$ to transfer to the opposite at a rate of 1 token/sec. But once the payments are executed as described in the initial state, the system throughput ends up at zero. This is because $C$ sends $B$ funds faster than its funds being replenished by $A$, reducing its transferable funds to zero, as shown in Fig. \ref{fig_PCN3}. However, $C$ needs positive funds to route transactions between $A$ and $B$. The transactions between $A$ and $B$ cannot be processed, even though they have sufficient funds. Thus the network goes into a \textit{deadlock} state.

\section{Splicer: Problem Statement}\label{sec_Problem}
\subsection{System Model}\label{sec_SM}
\textbf{Entities.} There are two types of entities in Splicer:
\begin{itemize}[leftmargin=*]
\item[$\bullet$] \textbf{Clients} are the end-users in PCNs who can send or receive a payment. We expect the clients to be lightweight, allowing mobile or IoT devices to make payments. Clients outsource the payment routing computation to smooth nodes. Each client interacts with a unique smooth node.

\item[$\bullet$] \textbf{Smooth nodes} process payment requests from clients by running routing protocol. Each smooth node makes path decisions for the current payment requests of its directly linked clients. Besides, multiple smooth nodes form a key management group (KMG) to create or retrieve keys with any distributed key generate protocol \cite{GRJ1999}.
\end{itemize}

\textbf{Topology.} Fig. \ref{PCH1} shows the star-like topology in state-of-the-art PCHs \cite{Ethan2017TumbleBit, Dziembowski2019Perun, Tairi}. A PCH opens payment channels with multiple clients separately. Thus, clients need one-hop routing for mutual payments through the intermediate PCH \cite{Rami2018CommitChains}. As shown in Fig. \ref{PCH2}, in Splicer we model the PCN with a \textit{multi-star-like topology}, in which the clients evenly connect to the smooth nodes (PCHs). Fig. \ref{PCH2} shows an example where the client consists of \textit{N} senders and \textit{M} recipients, and PCHs include \textit{Z} smooth nodes. We define the multi-star-like PCN topology as follows.
\begin{Def} \label{FedPCN}
	(Multi-star-like PCN topology). In a PCN, there are multiple PCHs connected directly or indirectly. The clients trade with each other through the PCHs that are directly connected to them.
\end{Def}
\noindent Notice that the payment hub model has been widely adopted in PCNs (e.g., \cite{Ethan2017TumbleBit, Tairi, Rami2018CommitChains, Dziembowski2019Perun}). However, we are the first to propose a multi-star-like PCN with multiple PCHs.

\begin{figure}[t]
	\vspace{-0.8cm}
	\centering
	\subfigure[Star-like topology.]{
		\begin{minipage}[t]{0.5\linewidth}
			\centering
			\includegraphics[width=1.5in]{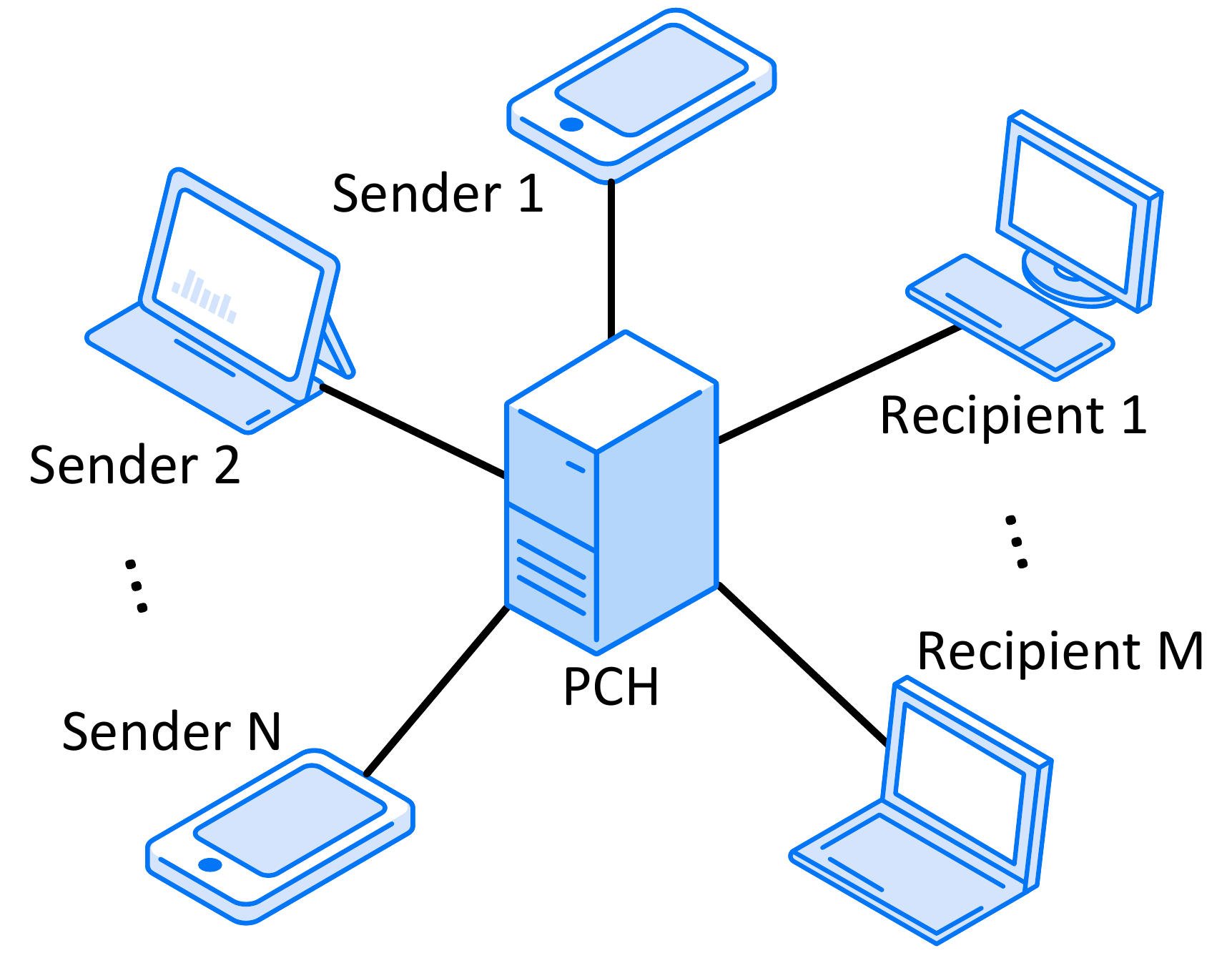}
			\label{PCH1}
		\end{minipage}%
	}%
	\subfigure[Multi-star-like topology.]{
		\begin{minipage}[t]{0.5\linewidth}
			\centering
			\includegraphics[width=1.5in]{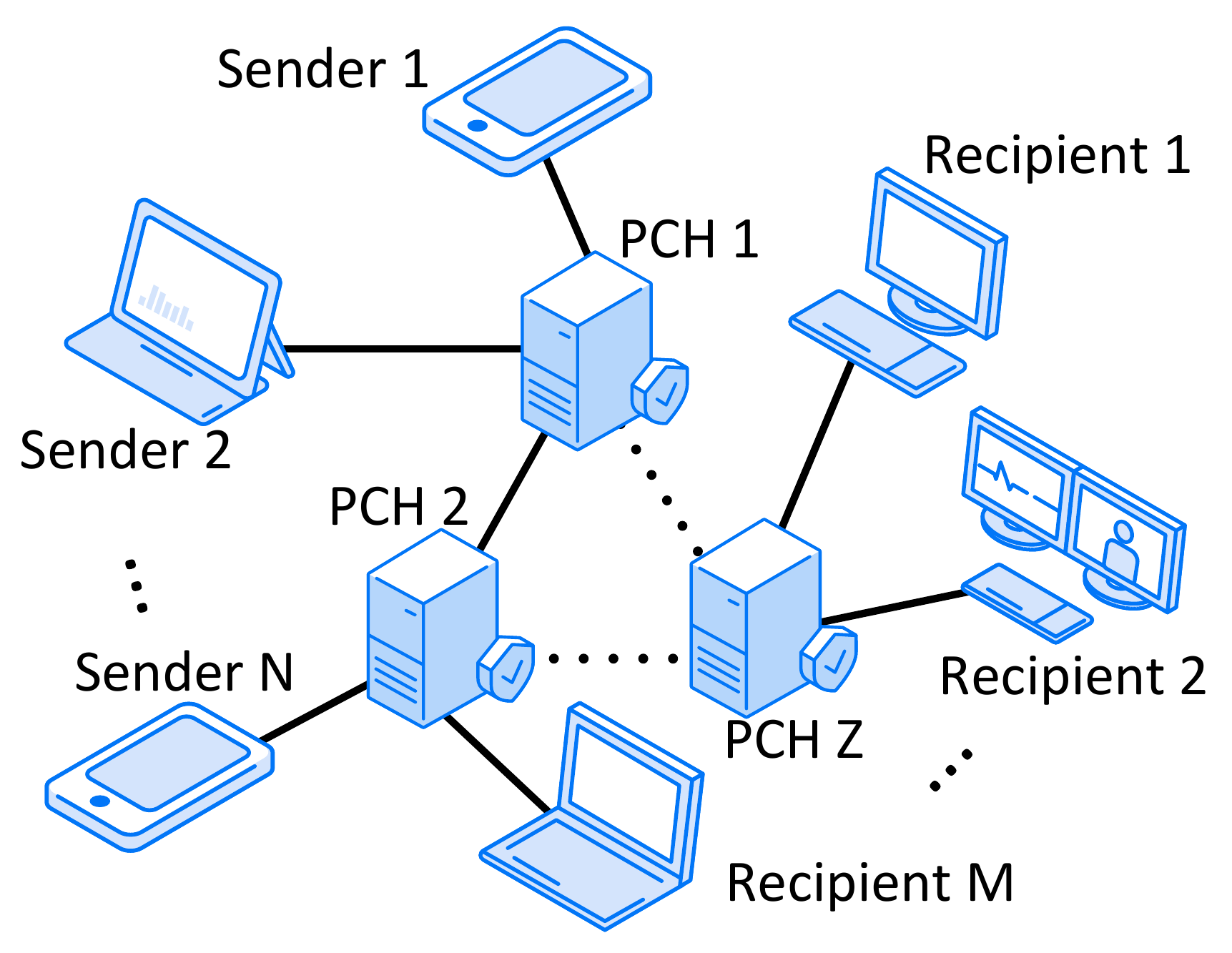}
			\label{PCH2}
		\end{minipage}
	}%
	\centering
	\vspace{-0.2cm}
	\caption{PCH topologies.}
	\vspace{-0.6cm}
\end{figure}

\textbf{Workflow.} A PCN can be modeled as a graph $\mathbb{G}=(\mathbb{V}, \mathbb{E})$ in which $\mathbb{V}$ denotes the set of nodes, and $\mathbb{E}$ denotes the set of payment channels between them. $\mathbb{V}_{\rm CLI} \subseteq \mathbb{V}$ denotes the clients, and $\mathbb{V}_{\rm SN} \subseteq \mathbb{V}$ denotes the smooth nodes, where $\mathbb{V}_{\rm CLI}=\mathbb{V}-\mathbb{V}_{\rm SN}$. As shown in Fig. \ref{workflow}, there is a payment from client $P_s$ to client $P_r$, $S_i$ and $S_j$ are the smooth nodes to which they are connected, for $P_s, P_r \in \mathbb{V}_{\rm CLI}$ and $S_i, S_j \in \mathbb{V}_{\rm SN}$. KMG contains $\iota$ smooth nodes, where $\iota$ is a system parameter. We now sketch the payment preparation and execution workflow.

\textit{a) Payment preparation:} For simplicity, we omit the detailed process of creating payment channels, and the channels' initial deposits are sufficient, similar to Ref. \cite{LiMZ20, Ge2022ShadufNP}. Now $P_s$ establishes a secure communication via transport layer security (TLS) protocol with $S_i$, so do $P_r$ and $S_j$. Then a payment channel is created between $P_s$ and $S_i$, so do $P_r$ and $S_j$. Next, $P_s$ initiates a payment request $\textsf{pay}_{\rm{req}}$ to $S_i$ via the secure channel so that $S_i$ knows the client has a new transaction to execute. Then $S_i$ starts the \textit{payment initialization}. $S_i$ creates a fresh transaction id \textsf{tid} and obtains fresh ($\textsf{pk}_{\textsf{tid}}$, $\textsf{sk}_{\textsf{tid}}$) pair from the KMG. $S_i$ sends \textsf{tid} and the corresponding public key  $\textsf{pk}_{\textsf{tid}}$ to $P_s$, and $S_i$ keeps $\textsf{sk}_{\textsf{tid}}$ private. Then $S_i$ generates the initial state $\textsf{state}_{\textsf{tid}}:=(\textsf{tid}, \theta_{\textsf{tid}})$, a tuple containing \textsf{tid} and a boolean $\theta_{\textsf{tid}}$ indicating whether the transaction is completed. 

\textit{b) Payment execution:} We illustrate the steps of payment execution in Fig. \ref{workflow} as follows:

\textbf{(1)} To start the process of executing a transaction \textsf{tid} with input \textsf{inp} which contains $D_{\textsf{tid}}$. Let $D_{\textsf{tid}}=(P_s, P_r, \textsf{val}_{\textsf{tid}})$ denote the payment demand of $P_s$, where $\textsf{val}_{\textsf{tid}}$ represents the payment amount. $P_s$ first computes $\textsf{inp}=\textsf{Enc}(\textsf{pk}_{\textsf{tid}}, D_{\textsf{tid}})$, then sends to $S_i$ a message $(\textsf{tid}, \textsf{inp})$ and the payment funds.

\textbf{(2-3)} $S_i$ decrypts $\textsf{inp}$ with $\textsf{sk}_{\textsf{tid}}$ to obtain $D_{\textsf{tid}} = \textsf{Dec}(\textsf{sk}_{\textsf{tid}}, \textsf{inp})$. Then the routing protocol splits the $D_{\textsf{tid}}$ into $K$ transaction-units (TUs) $D_{\textsf{tuid}}$ with fresh id \textsf{tuid}, generates the corresponding states $\textsf{state}_{\textsf{tuid}}^i =(\textsf{tuid}, \theta_{\textsf{tuid}}^{i})$, for $2\leq i\leq K$, and $\theta_{\textsf{tid}}=\bigwedge_{2\leq i\leq K}\theta_{\textsf{tuid}}^{i}$. From the KMG, $S_j$ obtains ($\textsf{pk}_{\textsf{tuid}}, \textsf{sk}_{\textsf{tuid}}$) pair. $S_i$ sends $\textsf{Enc}(\textsf{pk}_{\textsf{tuid}}, D_{\textsf{tuid}})$ to $S_j$, then $S_j$ decrypts it with $\textsf{sk}_{\textsf{tuid}}$. Once $S_j$ receives the funds corresponding to $D_{\textsf{tuid}}$, it returns to $S_i$ an acknowledgment $\textsf{ACK}_{\textsf{tuid}}$ via secure channel, with which $S_i$ updates the $\textsf{state}_{\textsf{tuid}}^i$ as $\theta_{\textsf{tuid}}^{i}=\textsf{true}$. When all acknowledgments $\textsf{ACK}_{\textsf{tuid}}$ are received, $S_i$ updates the $\textsf{state}_{\textsf{tid}}$ as $\theta_{\textsf{tid}}=\textsf{true}$.

\textbf{(4)} Finally, $S_j$ receives all TUs of $D_{\textsf{tid}}$, and sends the payment funds to $P_r$ at one time. $P_r$ generates a successful receipt acknowledgment $\textsf{ACK}_{\textsf{tid}}$, which is finally forwarded to $P_s$ by the smooth nodes.

In fact, in the above workflow, any transaction initiator needs to pay an additional forwarding fee to the intermediaries on the routing path, which is used as a forwarding incentive for smooth nodes, see \S \ref{protocol}.

\begin{figure}[t]
	\vspace{-0.8cm}
	\centering
	\includegraphics[width=3in]{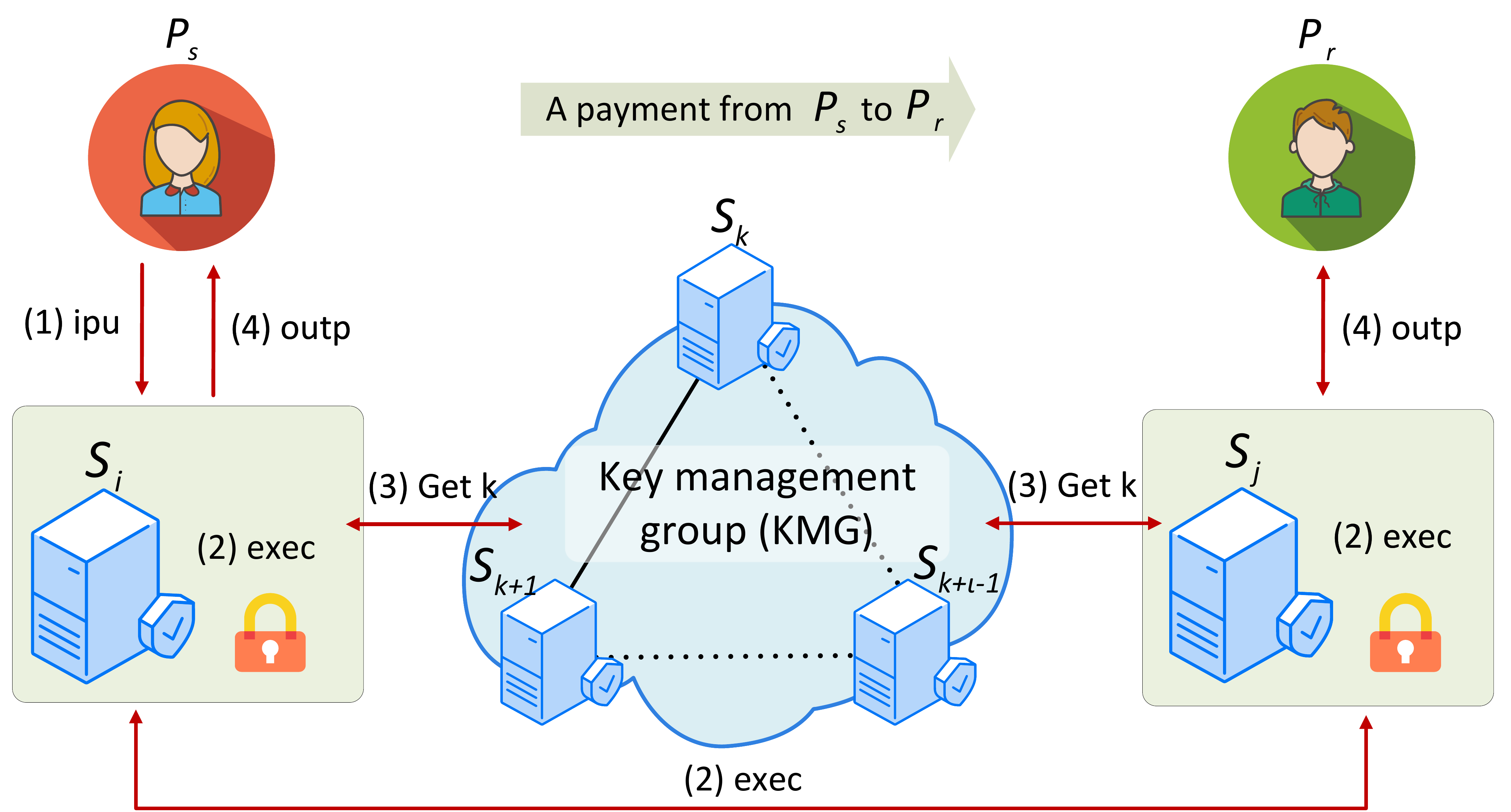}
	\vspace{-0.3cm}
	\caption{Workflow in multi-hop working model.}
	\label{workflow}
	\vspace{-0.7cm}
\end{figure}

\subsection{Trust, Communication and Threat Models}\label{privacy}
\textbf{Trust model.} Splicer is community-autonomous, and there is a certain trust transference between entities. In Fig. \ref{trust}, Splicer runs a multiwinner voting algorithm (e.g., \cite{CelisHV18}) in the smart contract that effectively allows all entities to fairly select a \textit{smooth node candidate list} in a long period. It can take into account the two properties: (i) \textit{Excellence} means the selected candidates are ``better" for outsourcing routing tasks (e.g., have more client connections, transaction funds, and lower operational overhead). (ii) \textit{Diversity} means that the candidate positions are as diverse as possible. We leave the optimal design of multiwinner voting for future work.

The first selected candidate list of smooth nodes temporarily performs payment routing as actual PCHs. When the network state is stable, the candidate smooth nodes run a smart contract containing a placement optimization algorithm to determine the actual PCHs (long term running). Notice that after the distribution of transaction requests is stabilized in the network, the overall distribution information of requests obtained by each candidate PCH is consistent, and the actual PCHs finally decided are consistent.

Fig. \ref{trust} shows that this process is accompanied by the removal of redundant payment channels, thus reducing the complexity of the network. Actual PCHs require pledging funds to a public pool for access, and their behavior checks and balances each other; if some PCHs appear malicious or colluding, they will be identified by other PCHs. Splicer also provides the client with a reporting and arbitration mechanism. The malicious PCHs will be removed, and their deposit will be confiscated as a punishment (the loss is greater than the profit). Then new PCHs will be selected from the new candidate list to supplement, so rational PCHs will not choose corruption. We emphasize that sensitive information (e.g., node identity, transaction content, routing data) is transmitted in ciphertext, coupled with unlinkability, so the client does not need to worry about privacy leakage.

\begin{figure}[t]
	\vspace{-0.8cm}
	\centering
	\includegraphics[width=3.3in]{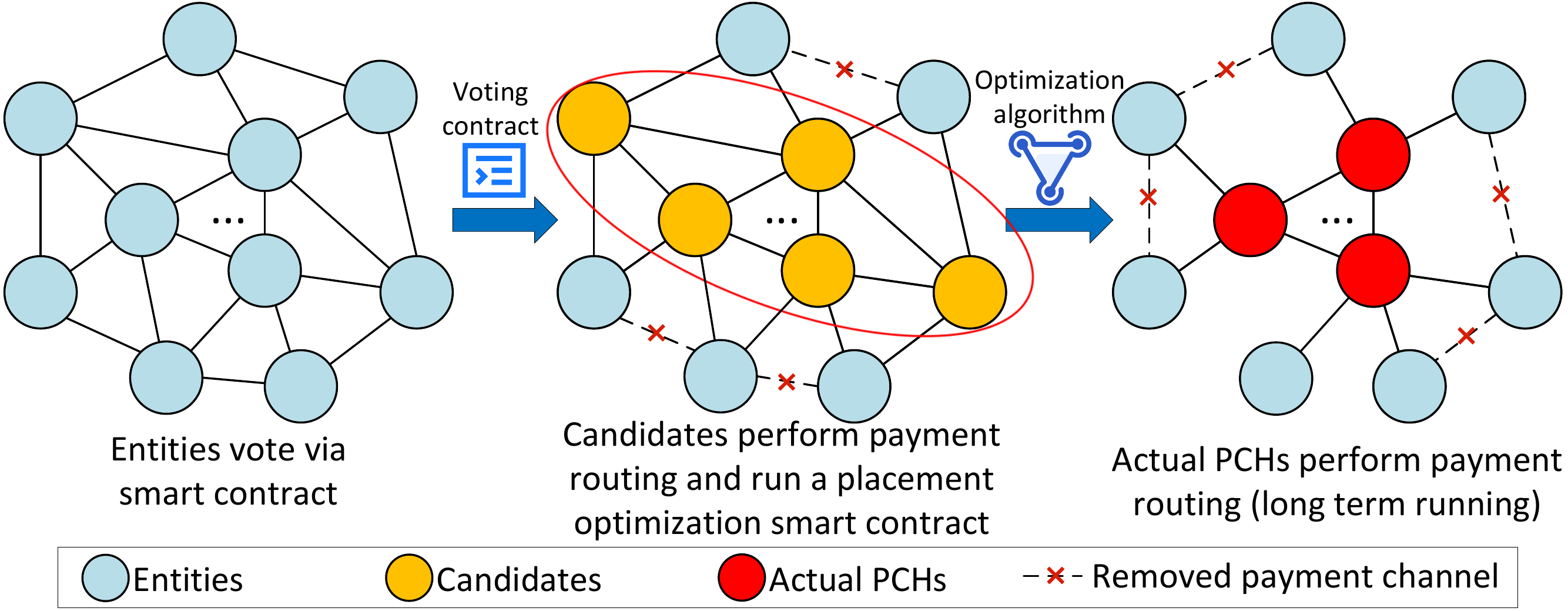}
	\vspace{-0.3cm}
	\caption{System trust transference model.}
	\label{trust}
	\vspace{-0.6cm}
\end{figure}

\textbf{Communication model.} As shown in Fig. \ref{com}, we sketch the communication process. Splicer runs under the bounded synchronous communication setting. At the beginning of epoch $e$+1, PCHs obtain and synchronize the final global information of the last epoch, including clients' states and network data (e.g., topology, channel state, payment flow rate, etc). Meanwhile,  after receiving directly connected clients' local payment requests, each PCH makes distributed routing decisions based on the network data (final global information of epoch $e$) and its clients' new requests (local information of epoch $e$+1). Finally, recipients generate the payment acknowledgments, which PCHs forward to senders. Splicer loops the above process.

\textbf{Threat model.} Each PCH is rational and potentially malicious, deviating from the protocol to obtain benefits. An adversary may compromise a target PCH's operating system and network stack, which can arbitrarily drop, delay, and replay messages. Since the adversary would not profit from corrupting the PCH placement process, the adversary's attack may only cause the payment routing of some transactions to fail. However, the failed transactions would be withdrawn by PCH without causing losses to the client or affecting the system's stability.

\vspace*{-0.5\baselineskip}
\subsection{PCH Placement and Routing Problems}\label{Sec_P}

\textbf{Placement problem.} The core of placement problem is to select the actual PCHs (fixed long-period running placement optimization smart contract) in the smooth node candidate list and make them install and run the PCH program for payment routing. The placement optimization smart contract determines the number of actual PCHs and their location in the network. We emphasize that this is a community-autonomous process and not a centralized decision. In long-term stable operation, the actual PCHs do not change, unless the network is not in an optimal operating state after long-term operation (the result of the placement optimization smart contract output changes) or a malicious PCH is removed. In practice, the community weighs the costs and benefits to determine when a new placement problem should be addressed.

\begin{figure}[t]
	\vspace{-0.8cm}
	\centering
	\includegraphics[width=3.3in]{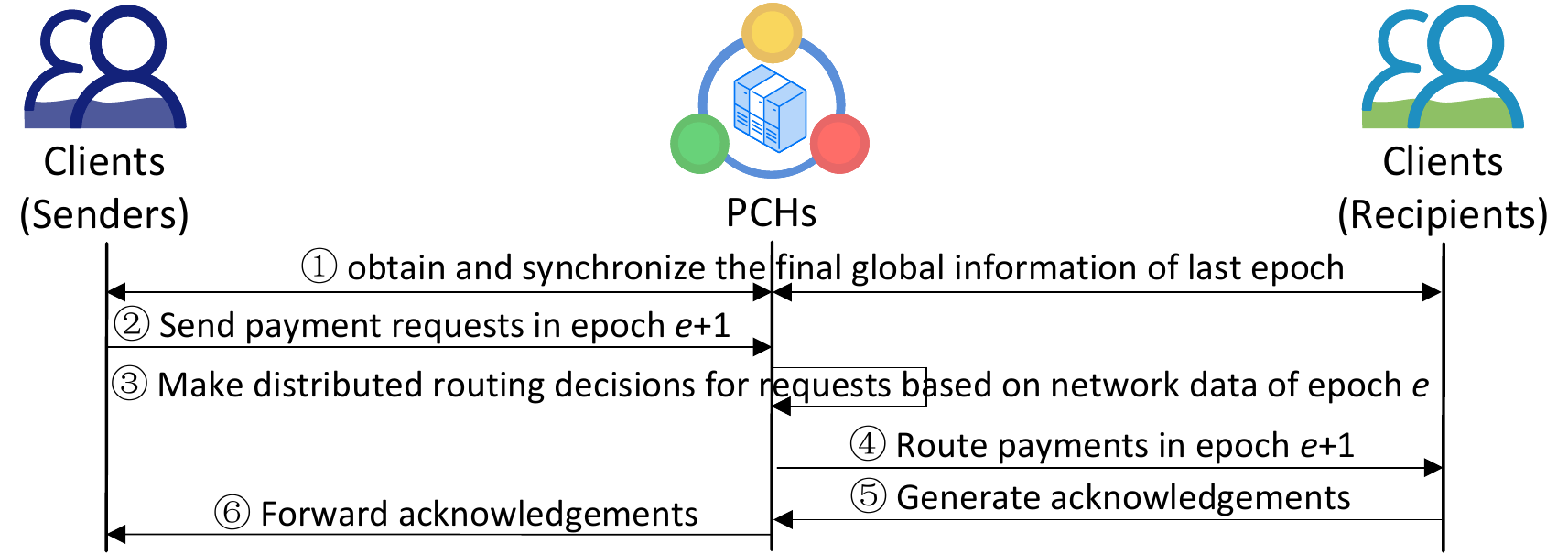}
	\vspace{-0.3cm}
	\caption{System communication process in epoch $e$+1.}
	\label{com}
	\vspace{-0.6cm}
\end{figure}

Since there are many geographically dispersed nodes in the real PCNs, PCHs may be far from some clients, leading to unstable connections or high communication delay and overhead between them. We aim to place the PCHs in proximity to the clients evenly, and all the clients have the lowest average payment hops forwarded by them. PCHs are physically distributed but logically polycentric and cooperative in managing payment routing. While such a placement strategy makes the distance between nodes shorter, it also considers the costs of PCHs collecting statistics from the clients and synchronizing between the PCHs. This creates a \textit{network load tradeoff}: (i) the PCHs should be near their routed clients to reduce the communication delay and route management overhead (\textbf{management cost}). (ii) they should be close to each other to reduce the delay and overhead of synchronizing states (\textbf{synchronization cost}). Thus the PCHs should be appropriately placed in a PCN, leading to a placement problem. We are the first to discuss the PCH placement problem in the PCNs, and we further describe the details of the placement problem in \S \ref{detail_problem} and the solutions in \S \ref{solutions}.

\textbf{Routing problem.} The existing routing solutions try to: (i) reduce routing costs to improve throughput and (ii) rebalance channel funds to improve routing performance. However, source routing requires each sender to compute the routing paths, which is limited in large-scale scenarios (Until April 11, 2023, the total number of nodes in the Lightning Network is 16,427. This paper defines a local network with more than 3,000 nodes as a large-scale network). We need to design a distributed routing decision protocol over multiple PCHs. Thus we propose a rate-based routing mechanism inspired by the ideas of packet-switching technology. Transactions are split into multiple independently routed TUs by each PCH. Each TU can transfer a few funds bounded by a \textit{Min-TU} and a \textit{Max-TU} value at different rates. We emphasize that this multi-path payment routing approach has proven to be feasible in Spider \cite{Sivaraman2020HighTC}. Notice that it does not affect the confidentiality of payments because each TU is encrypted with an independent public key from the KMG. In addition, the intermediate nodes of each TU routing path may be different, which confuses the relationship between two-party of multiple original transactions in PCNs. Thus, Splicer inherits the unlinkability of state-of-the-art PCHs. It is more difficult for intermediaries to identify sensitive information ciphertext. We elaborate on the rate-based routing protocol in \S \ref{protocol}.

\begin{figure*}[t]
	\vspace{-0.8cm}
	\centering
	\includegraphics[width=6.5in]{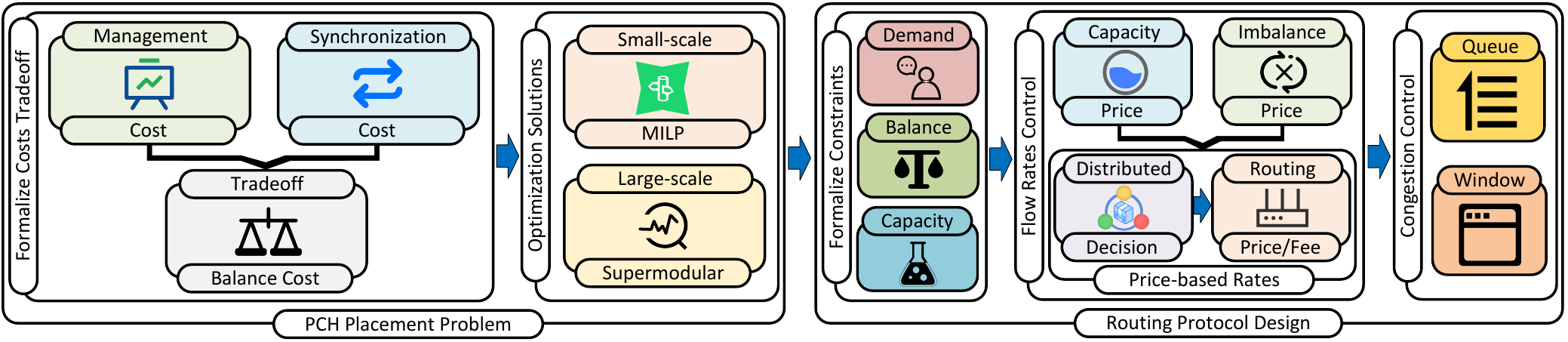}
	\vspace{-0.4cm}
	\caption{The overall structure of our design.}
	\label{system-overview}
	\vspace{-0.6cm}
\end{figure*}

\section{System Design}\label{sec_SD}

\subsection{Overview}
The overall structure of our design is shown in Fig. \ref{system-overview}. We first consider the placement of smooth nodes. The placement problem is to seek a tradeoff between routing management and synchronization costs. Now we study the details of the smooth nodes placement problem, formulating it as an optimization problem with two costs tradeoff and translating it into minimizing the \textit{balance cost}. Then, we provide solutions in two PCN scales for the transformed optimization problem. In a small-scale network, we transform the placement problem into a \textit{mixed-integer linear programming} (MILP) problem to find the optimal solution. We use supermodular function techniques to find the approximate solution in a large-scale network.

Next, we study the details of routing protocol design for smooth nodes. We first give the formal constraints of the routing problem, including \textit{demand}, \textit{capacity}, and \textit{balance} constraints. Then we consider the transaction flow rate control based on the routing price. We define the \textit{capacity and imbalance prices}, derive the \textit{routing price and fee} through distributed decisions, and obtain the price-based flow rates. Finally, we consider congestion control during routing and design the waiting queue and window to alleviate congestion.

\subsection{Formalize the Placement Problem} \label{detail_problem}
We describe the long periodic election of smooth node candidate list in the trust model of \S \ref{privacy}. We briefly present the PCH placement problem in \S \ref{Sec_P}. Next, we further model how to determine the actual PCHs in the candidate list.

To formally describe the network load tradeoff, we let two binary variables $x_n, y_{mn} \in \left\lbrace 0, 1 \right\rbrace$ denote whether a candidate node $n \in \mathbb{V}_{\rm SNC}$ ($\mathbb{V}_{\rm SNC}$ denotes the set of candidate smooth nodes) can be placed as a smooth node and whether a client $m \in \mathbb{V}_{\rm CLI}$ is assigned to the smooth node $n$, respectively. Thus, the vectors $\boldsymbol{x}$ and $\boldsymbol{y}$ show the placement and assignment plans, respectively:

\begin{small}
\begin{gather}
	\boldsymbol{x} = (x_n \in \left\lbrace 0, 1 \right\rbrace : n \in \mathbb{V}_{\rm SNC}), \label{constraint13} \\
	\boldsymbol{y} = (y_{mn} \in \left\lbrace 0, 1 \right\rbrace : m \in \mathbb{V}_{\rm CLI}, \, n \in \mathbb{V}_{\rm SNC}). \label{constraint14}
\end{gather}
\end{small}

A node $n$ is not capable enough to be placed as a smooth node ($x_n = 0, \forall n \notin \mathbb{V}_{\rm SNC}$). Each client needs to be assigned to a smooth node, and we require $\sum_{n \in \mathbb{V}_{\rm SNC}} y_{mn} = 1, \forall m \in \mathbb{V}_{\rm CLI}$. A node $n$ must be placed as a smooth node so that client $m$ can be assigned ($y_{mn} \leq x_n, \forall m \in \mathbb{V}_{\rm CLI}, \, n \in \mathbb{V}_{\rm SNC}$).

Let $\zeta_{mn}$ and $\delta_{nl}$ denote the management cost of assigning a client $m \in \mathbb{V}_{\rm CLI}$ to a smooth node $n \in \mathbb{V}_{\rm SNC}$, and the synchronization cost between two smooth nodes $n, l \in \mathbb{V}_{\rm SNC}$, respectively. Notice that $\zeta_{mn}$ and $\delta_{nl}$ are local or edge-wise parameters probed by candidate smooth nodes at the last long period. Then the total management cost and synchronization cost in the network can be expressed as


\begin{small}
\begin{gather}
\mathcal{C}_M(\boldsymbol{y}) = \sum_{m \in \mathbb{V}_{\rm CLI}}\sum_{n \in \mathbb{V}_{\rm SNC}} \zeta_{mn} y_{mn}, \label{cost1} \\
\mathcal{C}_S(\boldsymbol{x}, \boldsymbol{y}) = \sum_{n \in \mathbb{V}_{\rm SNC}}\sum_{l \in \mathbb{V}_{\rm SNC}} x_n x_l ( \delta_{nl} \sum_{m \in \mathbb{V}_{\rm CLI}} y_{mn} + \epsilon_{nl} ), \label{cost2}
\end{gather}
\end{small}

\noindent where $\epsilon_{nl}$ denotes the constant cost in synchronization.

The tradeoff is transformed into a balance between the costs shown in equations (\ref{cost1})-(\ref{cost2}). Let $\omega \ge 0$ denote the weight value between the two costs, and the balance cost can be stated as

\begin{small}
\begin{gather}
\mathcal{C}_B(\boldsymbol{x}, \boldsymbol{y}) = \mathcal{C}_M(\boldsymbol{y}) + \omega\mathcal{C}_S(\boldsymbol{x}, \boldsymbol{y}).
\end{gather}
\end{small}

The PCH placement problem is shown as $\min \mathcal{C}_B(\boldsymbol{x}, \boldsymbol{y})$, where the constraints are formulas (\ref{constraint13})-(\ref{constraint14}). The problem is complex in that it contains discrete variables and a nonlinear objective function (\ref{cost2}) with cubic and quadratic terms, and is a typical NP-hard problem \cite{QinPIT18}.

\subsection{Optimization Placement Problem Solutions}\label{solutions}

\textbf{Small-scale optimal solution.} We convert the placement problem to a MILP problem to find the small-scale optimal solution. The conversion is vital since after turning into a problem with a linear objective function with constraints, it can be solved easily by existing various commercial solvers.

We use standard linearization techniques to achieve this conversion process. First, we introduce two vectors $\boldsymbol{\vartheta}$ and $\boldsymbol{\varphi}$ as the additional optimization variables 

\begin{small}
\begin{gather}
\boldsymbol{\vartheta} = (\vartheta_{nl} \in \left\lbrace 0, 1 \right\rbrace : n,l \in \mathbb{V}_{\rm SNC}), \\
\boldsymbol{\varphi} = (\varphi_{nlm} \in \left\lbrace 0, 1 \right\rbrace : n,l \in \mathbb{V}_{\rm SNC}, \, m \in \mathbb{V}_{\rm CLI}).
\end{gather}
\end{small}

Second, the linear constraints for $\vartheta$ and $\varphi$ are as follows 

\begin{small}
\begin{gather}
\vartheta_{nl} \le x_n,\quad\vartheta_{nl} \le x_l,\quad\vartheta_{nl} \ge x_n + x_l - 1,\quad n,l \in \mathbb{V}_{\rm SNC}, \label{constraint_v2} \\
\begin{split}
\varphi_{nlm} \le \vartheta_{nl},\quad\varphi_{nlm} \le y_{mn},\quad\varphi_{nlm} \ge \vartheta_{nl} + y_{mn} - 1,\\
\quad n,l \in \mathbb{V}_{\rm SNC}, \, m \in \mathbb{V}_{\rm CLI}. \label{constraint_v4}
\end{split}
\end{gather}
\end{small}

\noindent Where the constraints in (\ref{constraint_v2}) mean that if at least one $x_n$ and $x_l$ are 0, $\vartheta_{nl}$ is 0; otherwise, it is 1. Similarly, the constraints in (\ref{constraint_v4}) work on the same principle.

Third, we linearize the cost function (\ref{cost2}) using the new variables, and it can be converted to

\begin{small}
\begin{gather}
\widehat{\mathcal{C}}_S(\boldsymbol{\vartheta}, \boldsymbol{\varphi}) = \sum_{n \in \mathbb{V}_{\rm SNC}}\sum_{l \in \mathbb{V}_{\rm SNC}} ( \sum_{m \in \mathbb{V}_{\rm CLI}}\delta_{nl}\varphi_{nlm} + \epsilon_{nl}\vartheta_{nl} ). 
\end{gather}
\end{small}

Finally, the MILP can be stated as $\min \mathcal{C}_M(\boldsymbol{y}) + \omega\widehat{\mathcal{C}}_S(\boldsymbol{\vartheta}, \boldsymbol{\varphi})$, where the constraints are formulas (\ref{constraint13})-(\ref{constraint14}) and (\ref{constraint_v2})-(\ref{constraint_v4}).

Therefore, the PCH placement problem has been converted to a MILP problem and can be directly solved by existing commercial solvers. The various solvers usually apply a combination of the branch and bound method and the cutting-plane method, which can solve the MILP problem quite fast for the small-scale problem. However, our model involves the payments of mobile or IoT devices, and the scale of the PCNs may be enormous, leading to an extremely large MILP problem, creating a bottleneck in the solvers' computational performance. Hence we overcome this problem by proposing an approximation solution to solve the large-scale problem. 

\textbf{Large-scale approximation solution.} Firstly, we introduce a lemma that reveals the relationship between the placement plan $\boldsymbol{x}$ and the assignment plan $\boldsymbol{y}$.

\begin{lemma} \label{lenmma1}
	Given a placement plan $\boldsymbol{x}$, for each $m \in \mathbb{V}_{\rm CLI}, n \in \mathbb{V}_{\rm SNC}$, the optimal assignment plan $\boldsymbol{y}$ can be expressed as

	\begin{small}
	\begin{gather} \label{eq_ymn}
	\hspace{-2mm}
	y_{mn} = 
	\left\{
	\begin{aligned}
	& 1, \,\, if \,\, n = \mathop{\arg\min}_{n' \in \mathbb{V}_{\rm SNC} : x_{n'} = 1} ( \omega\sum_{l \in \mathbb{V}_{\rm SNC} : x_l = 1} \delta_{n'l} + \zeta_{mn'} ), \\
	& 0, \,\, otherwise.
	\end{aligned}
	\right.
	\end{gather}
	\end{small}

\end{lemma}

\begin{proof}
	Assuming that there is an optimal assignment plan $\boldsymbol{y^o}$, in which the client $m^o$ is assigned to the smooth node $n^o$. Then there is another node $n^h \neq n^o$ and $n^h = 1$, let\\
	\begin{small}
	\begin{gather}
	\omega\sum_{l \in \mathbb{V}_{\rm SNC} : x_l = 1} \delta_{n^h l} + \zeta_{mn^h} < 
	\omega\sum_{l \in \mathbb{V}_{\rm SNC} : x_l = 1} \delta_{n^o l} + \zeta_{mn^o},
	\end{gather}
	\end{small}
	\noindent which shows that if the client $m^o$ is reassigned to the smooth node $n^h$, the management cost reduces $\zeta_{mn^o} - \zeta_{mn^h}$, and the synchronization cost reduces $\sum_{l \in \mathbb{V}_{\rm SNC} : x_l = 1} \delta_{n^o l} - \sum_{l \in \mathbb{V}_{\rm SNC} : x_l = 1} \delta_{n^h l}$. Thus the value of objective function $\mathcal{C}_B$ reduces $\omega\sum_{l \in \mathbb{V}_{\rm SNC} : x_l = 1} \delta_{n^o l} + \zeta_{mn^o} - \omega\sum_{l \in \mathbb{V}_{\rm SNC} : x_l = 1} \delta_{n^o l} - \zeta_{mn^o} > 0$. However, this contradicts our assumption that $\boldsymbol{y^o}$ is an optimal assignment plan.
\end{proof}

Lemma \ref{lenmma1} indicates that it is easy to find the assignment plan $\boldsymbol{y}$ for a given placement plan $\boldsymbol{x}$, thus we concentrate on optimizing the placement plan. Let $X_n$ represent the placement of a smooth node $n$ (i.e., $x_n = 1$), and the set of all possible placements of the smooth node is shown as

\begin{small}
\begin{gather}
\mathcal{S} = (X_n : n \in \mathbb{V}_{\rm SNC}),
\end{gather}
\end{small}

\noindent which means if and only if $X_n \in \mathcal{X}$, a subset $\mathcal{X} \subseteq \mathcal{S}$ shows a placement plan $\boldsymbol{x}$ that $x_n = 1$. Let $\boldsymbol{x}_\mathcal{X}$ denote the binary representation of $\mathcal{X}$, thus the balance cost objective function $\mathcal{C}_B$ can be denoted as a set of function $f: 2^{\mathcal{S}} \rightarrow \mathbb{R}$ 

\begin{small}
\begin{gather}
f(\mathcal{X}) = \mathcal{C}_B(\boldsymbol{x}_\mathcal{X}, y(\boldsymbol{x}_\mathcal{X})),
\end{gather}
\end{small}

\noindent where $y(\boldsymbol{x}_\mathcal{X})$ indicates the optimal assignment plan given the smooth node placement plan $\boldsymbol{x}_\mathcal{X}$ based on equation (\ref{eq_ymn}).

Secondly, we consider a well-researched class of set functions known as supermodular\cite{Ilev01}.
\vspace*{-0.3\baselineskip}
\begin{Def} \label{theprem1}
	Given a finite set $\mathcal{S}$, a set function $f: 2^{\mathcal{S}} \rightarrow \mathbb{R}$ is called supermodular if for all subsets $\mathcal{A}, \mathcal{B} \subseteq \mathcal{S}$ with $\mathcal{A} \subseteq \mathcal{B}$ and every element $i \in \mathcal{S} \setminus \mathcal{B}$ it holds that 
	\begin{small}
	\begin{gather}
	f(\mathcal{A} \cup \left\lbrace i \right\rbrace) - f(\mathcal{A}) \le f(\mathcal{B} \cup \left\lbrace i \right\rbrace) - f(\mathcal{B}),
	\end{gather}
	\end{small}
\end{Def}

\noindent which states that when an element $i$ is added to a set, the marginal value rises with the expansion of the respective set. 

\begin{lemma} \label{lemma2}
	The set function $f(\mathcal{X})$ is supermodular for the case of uniform costs $\delta_{nl} = \delta_{n'l'} = \delta$, $\forall n, l, n', l' \in \mathbb{V}_{\rm SNC}$.
\end{lemma}
The lemma \ref{lemma2} has been proved in \cite{QinPIT18}. Based on that, the placement problem can be molded as minimizing a supermodular function $f$.

\begin{algorithm}[!htbp]
	\DontPrintSemicolon
	\normalem 
	\caption{Placement Approximation Algorithm} \label{alg:random}
	\footnotesize{
		\KwInput{Two initially solutions $X_0^s$, $Y_0^s$, element $u_i$}
		\KwOutput{Final solution $X_{z}^s$ (or equivalently $Y_{z}^s$)}
		\For{$i = 1$ \textbf{to} $z$}
		{\label{A1l1}
			\tcp{Maintain the two solutions until they coincide}
			$a_i \leftarrow f(X_{i-1}^s \cup \left\lbrace u_i \right\rbrace) - f(X_{i-1}^s)$ \\
			$b_i \leftarrow f(Y_{i-1}^s \setminus \left\lbrace u_i \right\rbrace) - f(Y_{i-1}^s)$ \\
			$a_i' \leftarrow \max\left\lbrace a_i,0 \right\rbrace$, $b_i' \leftarrow \max\left\lbrace b_i,0 \right\rbrace$ \\
			\If{$a_i'/(a_i'+b_i')^{\star}$} 
			{ \label{A1l5}
				$X_i^s \leftarrow X_{i-1}^s \cup \left\lbrace u_i \right\rbrace$, $Y_{i}^s \leftarrow Y_{i-1}^s$ \\
			}	
			\Else
			{
				$X_i^s \leftarrow X_{i-1}^s$, $Y_{i}^s \leftarrow Y_{i-1}^s \setminus \left\lbrace u_i \right\rbrace$ \\
			}
		}
		\uline{\textbf{return}} $X_{z}^s$ (or equivalently $Y_{z}^s$) \label{A1l9}\\
		\vspace{0.5mm}
		$^{\star}$ If $a_i'=b_i'=0$, then $a_i'/(a_i'+b_i')=1$ \label{A1l10}
	}
\end{algorithm}

Thirdly, solving this kind of problem is equivalent to addressing their submodular function maximization version. Let $f^{ub}$ denote an upper limit of the highest possible value of $f(\mathcal{X})$, the submodular function is 
$\widehat{f}(\mathcal{X}) = f^{ub} - f(\mathcal{X})$.

There are various approximation algorithms (e.g., \cite{FeldmanNS11, BuchbinderFNS15}) to maximize $\widehat{f}(\mathcal{X})$, and an approximation bound $\psi$ indicates the ratio of the value of the approximate solution over the optimal solution value is always at least $\psi$. The algorithm in \cite{BuchbinderFNS15} provides the best approximation bound, which  $\psi = \frac{1}{2}$. As is outlined in Alg. \ref{alg:random}, it yields in $z = |\mathbb{V}_{\rm SNC}|$ iterations, and $u_i(1 \le i \le z)$ is an arbitrary element of set $\mathcal{S}$. Two solutions $X^s$ and $Y^s$ initially set as $X_0^s \leftarrow \varnothing$ and $Y_0^s \leftarrow \mathcal{S}$. \textbf{Lines \ref{A1l1}-\ref{A1l9}} show as follows. At the $i^{th}$ iteration, the algorithm adds $u_i$ to $X_{i-1}^s$ or removes $u_i$ from $Y_{i-1}^s$ randomly and greedily based on the marginal gain of each of the two options. Thus the algorithm generates two random solutions $X_i^s$ and $Y_i^s$. After $z$ iterations, both solutions coincide (i.e., $X_z^s = Y_z^s$), and it is returned in line \ref{A1l9}. Line \ref{A1l10} handles a special case in line \ref{A1l5} where $a_i'=b_i'=0$. Lastly, based on the above, we can get the approximate solution of the large-scale network instances.

\subsection{Rate-Based Routing Protocol Design} \label{protocol}

\textbf{The formal constraints.} For a path $p$, let $r_p$ represent the payment rate sent on $p$ from the start to the end. We assume that TUs are sent through a payment channel of capacity $c_{a,b}$ from the smooth node $a$ to another smooth node $b$ at a rate $r_{a,b}$. Once payment is forwarded, it takes $\Delta$ time on average to receive the TUs acknowledgment from the end, thus $r_{a,b}\Delta$ funds are locked in the channel. The \textit{capacity constraint} on the channel means that the average rate cannot exceed $c_{a,b}/\Delta$. In addition, in order to ensure the channel fund balance, there is a \textit{balance constraint} that the one direction payment rate $r_{a,b}$ needs to match the rate $r_{b,a}$ in the reverse direction. Otherwise, the funds move to one end of the channel, and eventually, all converge at one end, creating a local deadlock (see \S \ref{deadlock}).

To ensure the full utilization of funds in channels, we consider a common model of \textit{utility} for making payments. The logarithm of the total rate at which payments are sent from a source represents the utility of the source \cite{KellyV05}. Therefore, we seek to maximize the total utility of whole source-destination pairs subject to the above constraints as follows:

\begin{small}
\begin{align}
\max \quad &\sum_{s,e \in \mathbb{V}} log ( \sum_{p \in \mathbb{P}_{s,e}} r_p ) \\
s.t. \quad &\sum_{p \in \mathbb{P}_{s,e}} r_p\Delta \leq d_{s,e} \quad\  \forall s,e \in \mathbb{V} \label{demand}\\ 
&\  r_{a,b} + r_{b,a} \leq \frac{c_{a,b}}{\Delta} \quad \forall (a, b) \in \mathbb{E} \label{capacity}\\
&\  \lvert r_{a,b} - r_{b,a} \rvert \leq \epsilon \ \ \ \ \, \, \forall (a, b) \in \mathbb{E} \label{balance}\\
&\  r_p \geq 0 \qquad\qquad\ \ \ \ \, \forall p \in \mathbb{P},
\end{align}
\end{small}

\noindent where $s$ denotes the start and $e$ denotes the end, $\mathbb{P}_{s,e}$ is the set of all paths from $s$ to $e$, $d_{s,e}$ is the demand from $s$ to $e$. $c_{a, b}$ denotes the capacity of the channel $(a,b)$, and $\mathbb{P}$ denotes the set of all paths. Formula (\ref{demand}) indicates the \textit{demand constraint}, ensuring the total flow of all paths is no more than the total demand. Formula (\ref{capacity}) and (\ref{balance}) are \textit{capacity} and \textit{balance constraints}, respectively. The balance constraint is harsh in the ideal case (i.e., the system parameter $\epsilon=0$), but we intend that the flow rates in both channel directions tend towards equilibrium in practice (i.e., $\epsilon$ is small enough).

\textbf{Distributed routing decisions.} Each PCH makes distributed routing (incremental) decisions for payments based on the network data of the last epoch and its clients' requests in the current epoch. Based on primal-dual decomposition techniques \cite{kelly2005stability}, we consider the optimization problem for a generic utility function $U(\sum_{p \in \mathbb{P}_{s,e}} r_p)$. Lagrangian decomposition can naturally decompose this linear programming problem into separate subproblems \cite{palomar2006tutorial}. A solution is to compute the flow rates that should be maintained on each path. We set the \textit{routing price} in both directions of each channel, and the PCHs adjust the prices to control the flow rates of TUs. Meanwhile, the routing price is used as the \textit{forwarding fee} to incentivize PCHs. 

The routing protocol is shown in Alg. \ref{alg:RP}. \textbf{(Lines \ref{A2l1}-\ref{A2l7})} There is decrypted a payment demand $D_{s,e}$, the smooth node splits it into $k$ packets of TUs $d_i$ (we limit $\textit{Min-TU}  \leq |d_i| \leq \textit{Max-TU}$ to control the number of split TUs), $|D_{s,e}| = \sum_{i=1}^{k} |d_i|$, and there are $k$ paths $\{p_i\}_{1 \leq i \leq k} \in \mathbb{P}_{s, e}$ (see \S \ref{Choices} for a discussion of choosing different paths). For brevity, we only consider a channel $(a, b)$ in path $p_i$, $(a, b) \in p_i$. Let $\lambda_{a, b}$ denote the \textit{capacity price} that indicates the total rate of arrival transactions exceeds the capacity, and let $\mu_{a, b}$ and $\mu_{b, a}$ denote the \textit{imbalance price} that represents the imbalance of rate in the two directions, respectively. These three prices are updated every $\tau$ seconds to keep the capacity and balance constraints from being violated. Let $n_a, n_b$ denote the funds required to maintain the flow rates at $a$ and $b$. The capacity price $\lambda_{a, b}$ is updated as

\begin{small}
\begin{gather}
\lambda_{a, b}(t+1) = \lambda_{a, b}(t) + \kappa(n_a(t) + n_b(t) - c_{a, b}) \label{price1},
\end{gather}
\end{small}

\noindent where $\kappa$ is a system parameter used to control the rate of price change. Any required funds excessing the capacity $c_{a, b}$ cause the capacity price $\lambda_{a, b}$ to rise, which indicates that the rates via $a, b$ need to be reduced and vice-versa. 

Let $m_a, m_b$ represent the TUs arriving at $a$ and $b$ in the last period, respectively. The imbalance price $\mu_{b, a}$ is updated as 

\begin{small}
\begin{gather}
\mu_{a, b}(t+1) = \mu_{a, b}(t) + \eta(m_a(t) - m_b(t)) \label{price2},
\end{gather}
\end{small}

\noindent where $\eta$ is a system parameter. Any funds arriving in $(a, b)$ direction more than $(b, a)$ direction cause the imbalance price $\mu_{a, b}$ to increase and $\mu_{b, a}$ to decrease, which means that the rates routing along $(a, b)$ need to be throttled and vice-versa.  

Based on observations of routing prices and node feedbacks, the smooth nodes run a multi-path routing protocol to control the rates at which payments are transferred. Probes \cite{WangXJW19} are sent periodically every $\tau$ seconds on each path to measure the above two prices. The routing price of the channel $(a, b)$ is

\begin{small}
\begin{gather}
\xi_{a, b} = 2\lambda_{a, b} + \mu_{a, b} - \mu_{b, a} \label{price3},
\end{gather}
\end{small}

\noindent the forwarding fee that $a$ needs to pay to $b$ is

\begin{small}
\begin{gather}
	\textsf{fee}_{a, b} = T_{\textsf{fee}}(2\lambda_{a, b} + \mu_{a, b} - \mu_{b, a}) \label{fee},
\end{gather}
\end{small}

\noindent where $T_{\textsf{fee}} (0<T_{\textsf{fee}}<1)$ is a systematic threshold parameter.

Thus, the total routing price of a path $p$ is

\begin{small}
\begin{gather}
\varrho_p = (1+T_{\textsf{fee}})\sum_{(a, b) \in p} \xi_{a, b} \label{price4},
\end{gather}
\end{small}

\noindent which indicates the total amount of excess and imbalance demands. Then the smooth node sends a probe on path $p$, which sums the price $\xi_{a, b}$ of each channel $(a, b)$ on $p$. Based on the routing price $\varrho_p$ from the most recently received probe, the rate $r_p$ is updated as

\begin{small}
\begin{gather}
r_p(t+1) = r_p(t) + \alpha(U'(r) - \varrho_p(t)) \label{r1},
\end{gather}
\end{small}

\noindent where $\alpha$ is a system parameter. Therefore, the sending rate on a path is adjusted reasonably according to the routing price.

\begin{algorithm}[!htbp]
	\DontPrintSemicolon
	\normalem 
	\caption{Distributed Routing Decision Protocol} \label{alg:RP}
	\footnotesize{
	\KwInput{Decrypted demand $D_{s,e}$, rate $r_{p_i}$, required funds $n_a$, $n_b$, arrived TUs $m_a$, $m_b$}
	\KwOutput{Routing rates $r_{p_i}$ $(1 \leq i \leq k)$}
	Split the demand $D_{s,e}$ into $d_i$ on path $p_i$ $(1 \leq i \leq k)$. \label{A2l1}\\
	\For{$i = 1$ \textbf{to} $k$}
	{
		\For{$\forall$ payment channel $(a, b)$ on path $p_i$}
		{
			\tcp{Update the routing rates}
			$\lambda_{a, b} \leftarrow \lambda_{a, b} + \kappa(n_a + n_b - c_{a, b})$ \label{A2l3}\\
			$\mu_{a, b} \leftarrow \mu_{a, b} + \eta(m_a - m_b)$\\
			$\xi_{a, b} \leftarrow 2\lambda_{a, b} + \mu_{a, b} - \mu_{b, a}$\\
			\tcp{Forwarding fee}
			$\textsf{fee}_{a, b} \leftarrow T_{\textsf{fee}}(2\lambda_{a, b} + \mu_{a, b} - \mu_{b, a})$ \\
			\tcp{Routing price of the path $p_i$}
			$\varrho_{p_i} \leftarrow (1+T_{\textsf{fee}})\sum_{(a, b) \in p_i} \xi_{a, b}$\\
			$r_{p_i} \leftarrow r_{p_i} + \alpha(U'(r) - \varrho_{p_i})$ \label{A2l7}\\
			\tcp{Congestion control}
			\If{$r_{p_i} > r_{a, b}^{\rm process}$ \textbf{or} $F_{a, b} < |d_i|$}
			{\label{A2l8}
				$q_{a, b}^{\rm amount} \leftarrow d_i$\\
				$t_{p_i}^{\rm delay}$ $\leftarrow$ Smooth nodes monitor\\
				\If{$t_{p_i}^{\rm delay} > T$}
				{
					$d_i^* \leftarrow d_i$ \label{A2l12}\\
				}
				\If{$d_i^*$ is aborted}
				{\label{A2l13}
					$w_{p_i} \leftarrow w_{p_i} - \beta$\\
				}
				\If{$q_{a, b}^{\rm amount} < w_{p_i}$ \textbf{and} $d_i$  is transmitted}
				{
					$w_{p_i} \leftarrow w_{p_i} + \frac{\gamma}{\sum_{p_i' \in \mathbb{P}_{s,e}}w_{p_i'}}$ \label{A2l16}\\
				}
			}
		}
		\textbf{return} $r_{p_i}$	
	}
	}
\end{algorithm}

\textbf{Congestion control (Lines \ref{A2l8}-\ref{A2l16}).} We consider that this rate-based approach may cause congestion of TUs, so we use the waiting queue and window to control congestion. Whenever congestion occurs, intermediate hubs in the path queue up the TUs, representing a capacity or balance constraint violation. So the smooth nodes need to use a congestion control protocol to detect capacity and imbalance violations to control queues by adjusting the sending rates in channels.  

The congestion controller has two basic properties to achieve both efficiency and balanced rates. (i) It should try to keep the queue not empty, which indicates that the channel capacity is being used efficiently. (ii) It should keep the queues bounded, which means that the flow rate of each path can not exceed capacity or be imbalanced. There are some congestion control algorithms \cite{HaRX08} that satisfy the two properties and can be adapted for PCNs. Now we describe the protocol briefly:

If the rate $r_{p_i}$ exceeds the upper limited rate $r_{a, b}^{\rm process}$ that the channel $(a, b)$ can process, or the demand $|d_i|$ exceeds the current funds in $(a \rightarrow b)$ direction $F_{a, b}$, then the protocol goes to the congestion control part. \textit{\textbf{(i)}} Let $q_{a, b}^{\rm amount}$ denote the amount of TUs pending in queue $q_{a, b}$. The queuing delay $t_{p}^{\rm delay}$ on path $p$ is monitored by smooth nodes, and if it exceeds the pre-determined threshold $T$, the packet $d_i$ is marked as $d_i^*$. Once a TU is already marked, hubs do not process the packet and merely forward it. When the recipient sends back an acknowledgment with the marked field appropriately set, hubs forward it back to the sender. \textit{\textbf{(ii)}} Based on the observations of congestion in the network, smooth nodes control the payments rates transferred in the channels and choose a set of $k$ paths to route TUs from $s$ to $e$. \textit{\textbf{(iii)}} The window size $w_p$ represents the maximum of unfinished TUs on path $p$. The smooth nodes maintain the window size for every candidate path to a destination, which indirectly controls the flow rate of TUs on the path. The smooth nodes keep track of unserved or aborted TUs on the paths. New TUs can be transmitted on path $p$ only if the total number of TUs to be processed does not exceed $w_p$. On a path $p$ from $s$ to $e$, the window is adjusted as

\begin{small}
\begin{gather}
w_p \left( t+1\right) = w_p \left( t\right) - \beta, \label{w1}\\ 
w_p \left( t+1\right) = w_p \left( t\right) + \frac{\gamma}{\sum_{p' \in \mathbb{P}_{s,e}}w_{p'}}, \label{w2}
\end{gather}  
\end{small} 

\noindent where equation (\ref{w1}) means the marked packets fail to complete the payment within the deadline, and the senders choose to cancel the payment, and equation (\ref{w2}) means the unmarked packets are transmitted. The positive constants $\beta$ and $\gamma$ denote the factors that the window size decreases and increases.

Notice that although Spider \cite {Sivaraman2020HighTC} uses a similar multi-path payment model, the main \textit{differences} of Splicer are: (i) Splicer considers forwarding costs, though the fee model is different from that of the Lightning Network in Spider. (ii) Besides congestion control, Splicer also provides rate control as to minimize the capacity and imbalance violations in the network. (iii) Splicer's route computation is outsourced to PCHs instead of being processed by end-users. In the next section, we evaluate the performance of such an optimized Splicer vs. Spider solution.

\begin{figure*}[t]
	\vspace{-8mm}
	\centering
	\subfigure[Influence of the channel size]{
		\begin{minipage}[t]{0.23\linewidth}
			\centering
			\includegraphics[width=1.4in]{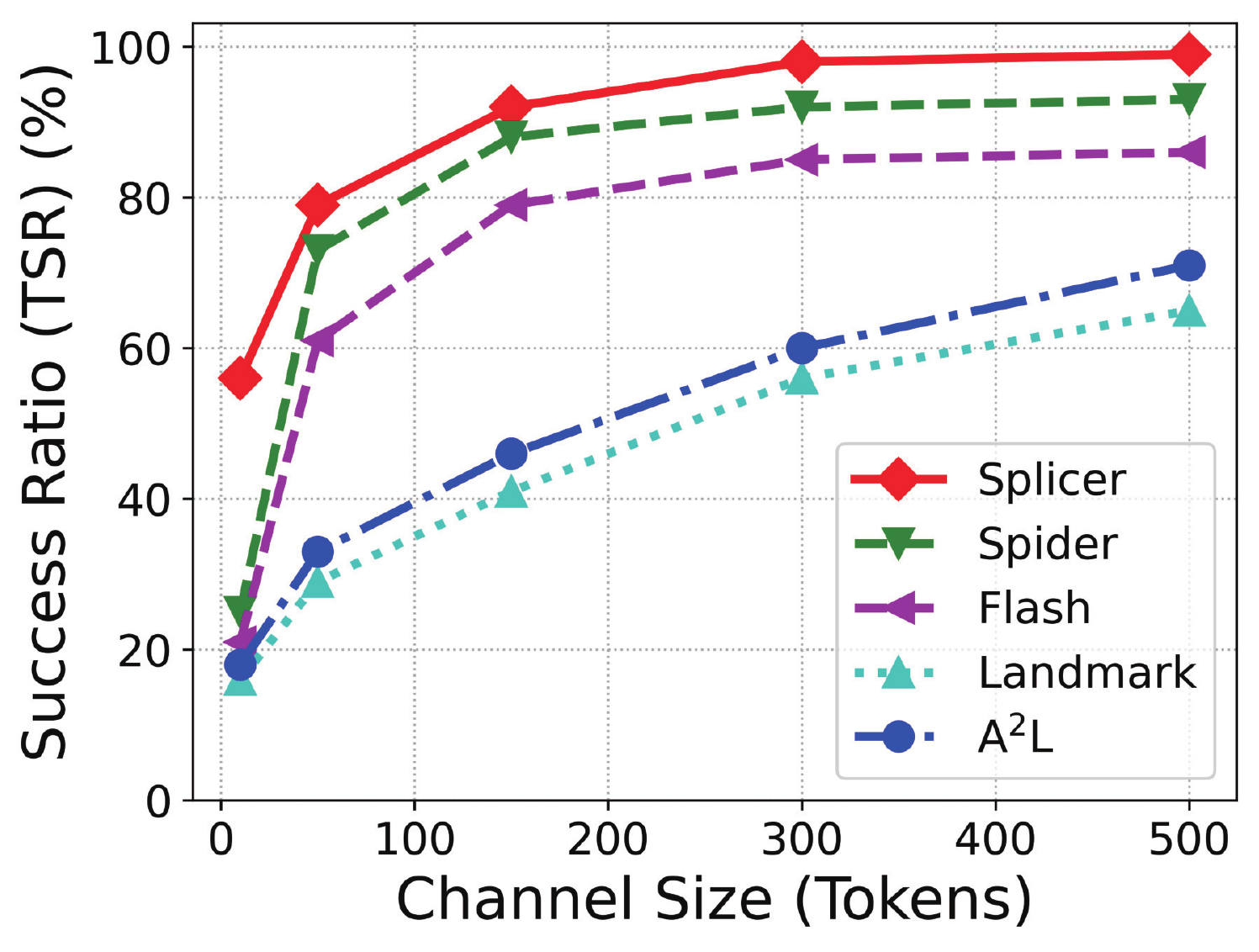}
			\label{fig_channelSize-successRatio}
			\vspace{-1cm}
		\end{minipage}%
	}%
	\hspace{1mm}
	\subfigure[Influence of the transaction size]{
		\begin{minipage}[t]{0.23\linewidth}
			\centering
			\includegraphics[width=1.4in]{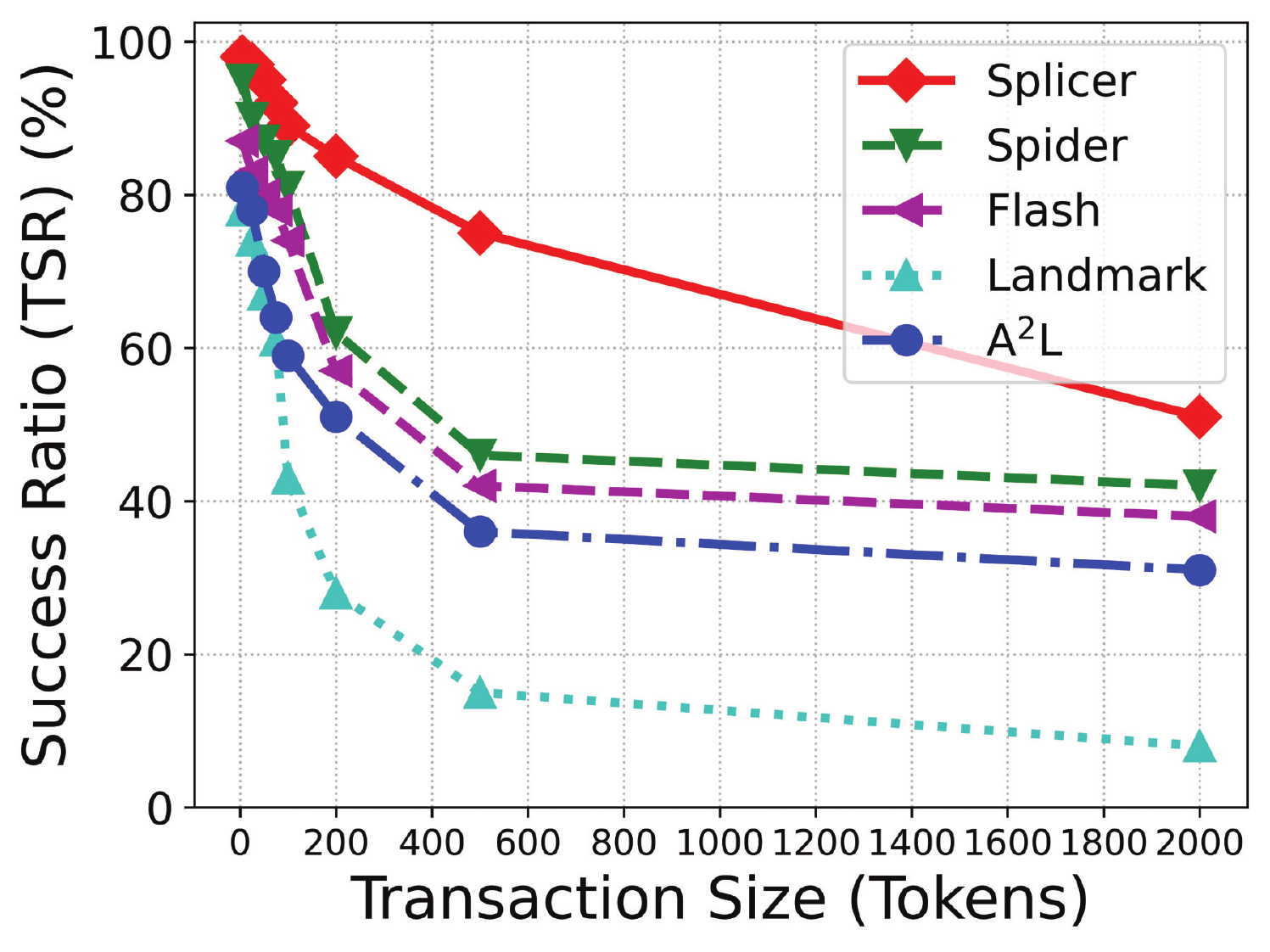}
			\label{fig_TransactionSize-successRatio}
			\vspace{-1cm}
		\end{minipage}%
	}%
	\hspace{1mm}
	\subfigure[Influence of the update time]{
		\begin{minipage}[t]{0.23\linewidth}
			\centering
			\includegraphics[width=1.4in]{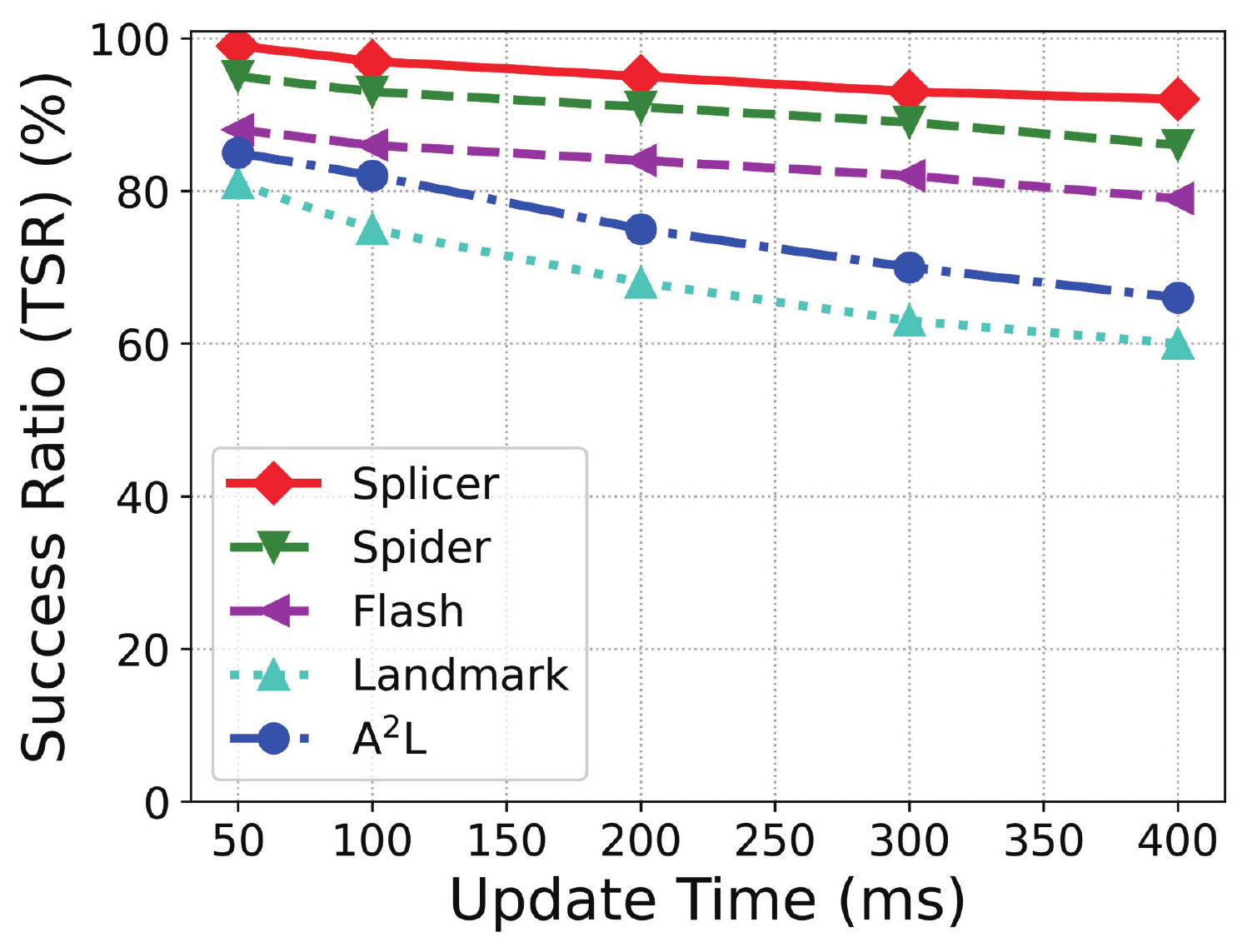}
			\label{fig_updateTime-successRatio}
			\vspace{-1cm}
		\end{minipage}
	}%
	\hspace{1mm}
	\subfigure[Normalized throughput]{
		\begin{minipage}[t]{0.23\linewidth}
			\centering
			\includegraphics[width=1.4in]{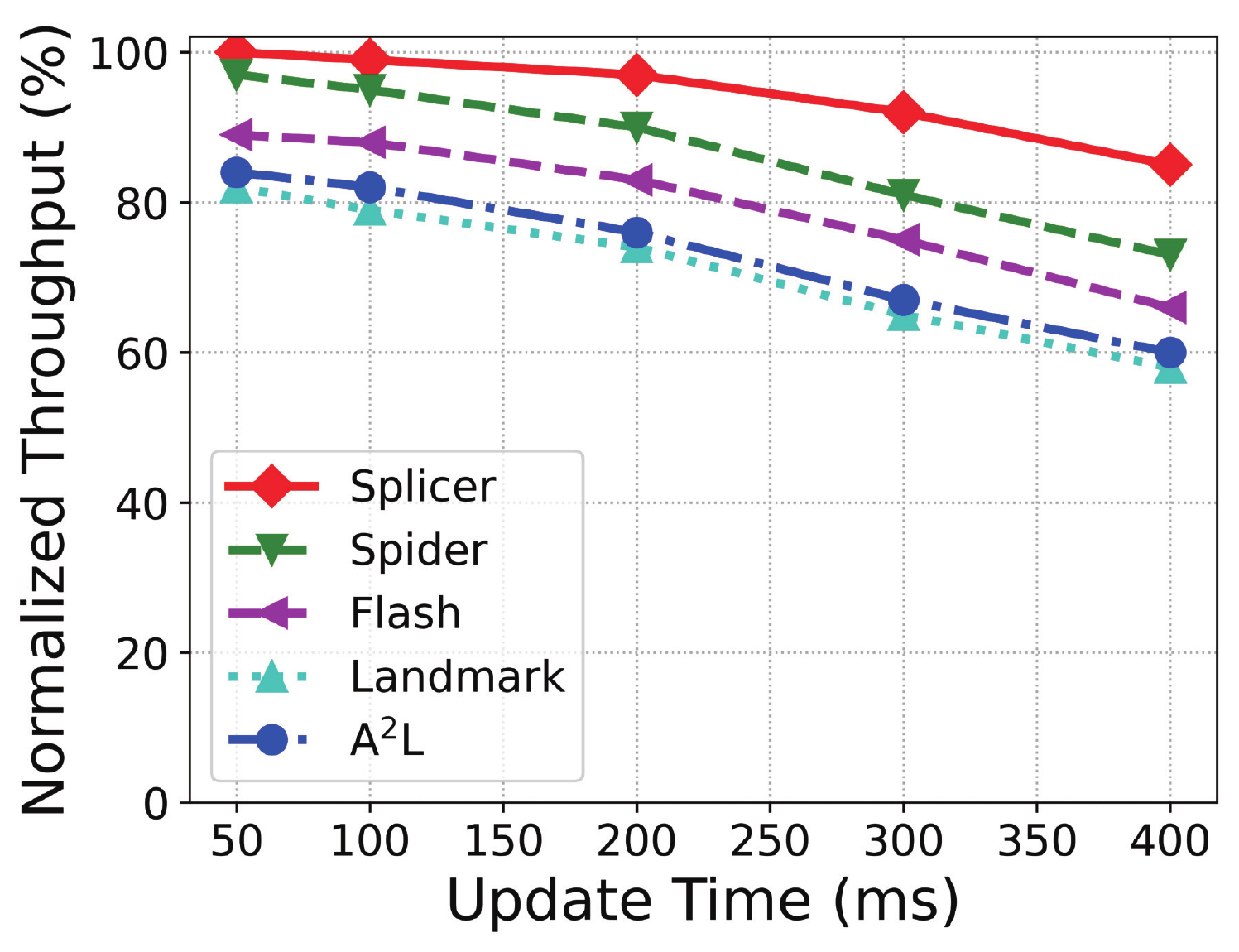}
			\label{fig_updateTime-throughput}
			\vspace{-1cm}
		\end{minipage}
	}%
	\centering
	\vspace{-0.3cm}
	\caption{The comparison between Splicer and other schemes under different metrics in small-scale networks.}
	\label{performance}
	\vspace{-0.4cm}
\end{figure*}

\begin{figure*}[t]
	\centering
	\subfigure[Influence of the channel size]{
		\begin{minipage}[t]{0.23\linewidth}
			\centering
			\includegraphics[width=1.4in]{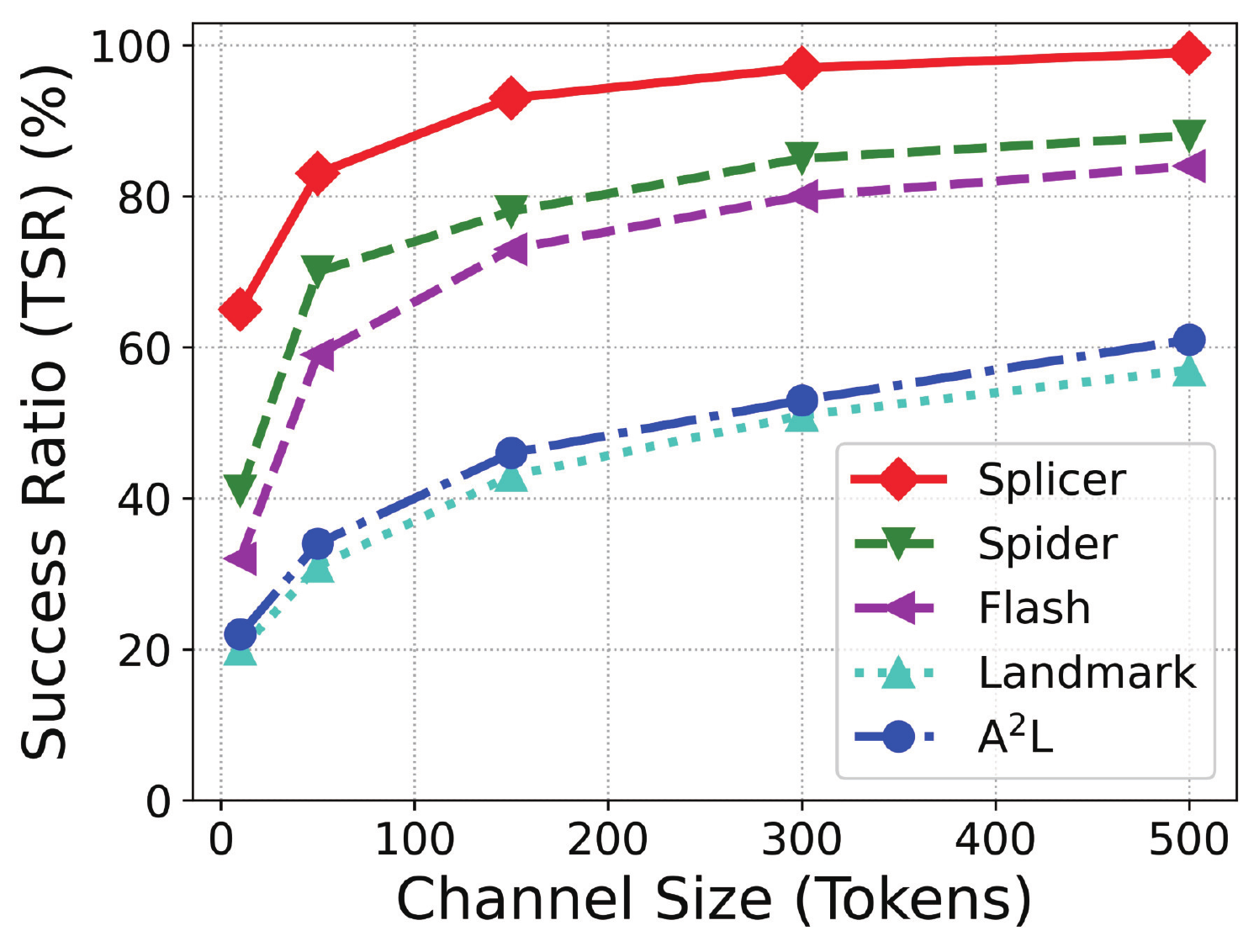}
			\label{fig_channelSize-successRatio2}
			\vspace{-1cm}
		\end{minipage}%
	}%
	\hspace{1mm}
	\subfigure[Influence of the transaction size]{
		\begin{minipage}[t]{0.23\linewidth}
			\centering
			\includegraphics[width=1.4in]{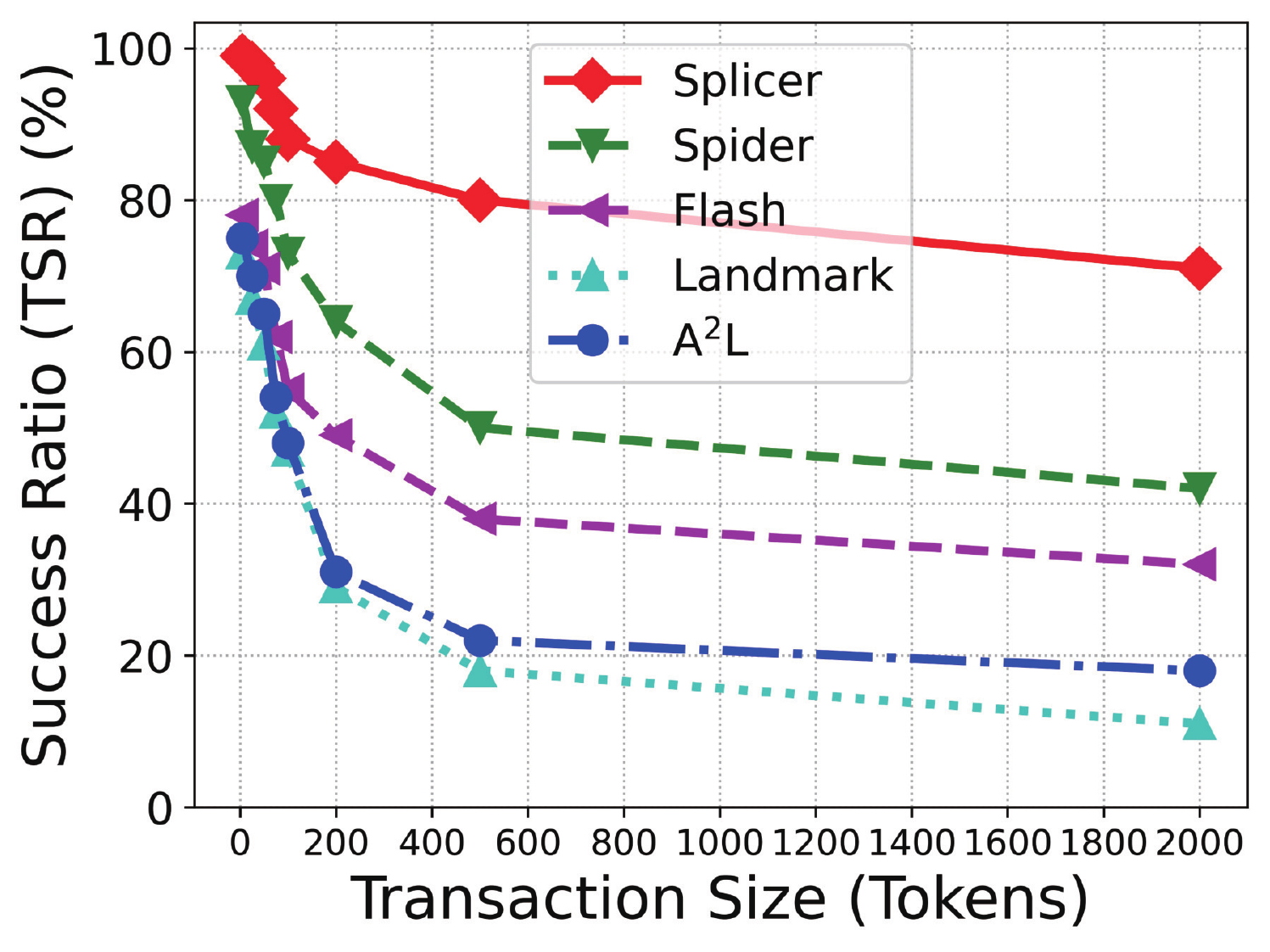}
			\label{fig_TransactionSize-successRatio2}
			\vspace{-1cm}
		\end{minipage}%
	}%
	\hspace{1mm}
	\subfigure[Influence of the update time]{
		\begin{minipage}[t]{0.23\linewidth}
			\centering
			\includegraphics[width=1.4in]{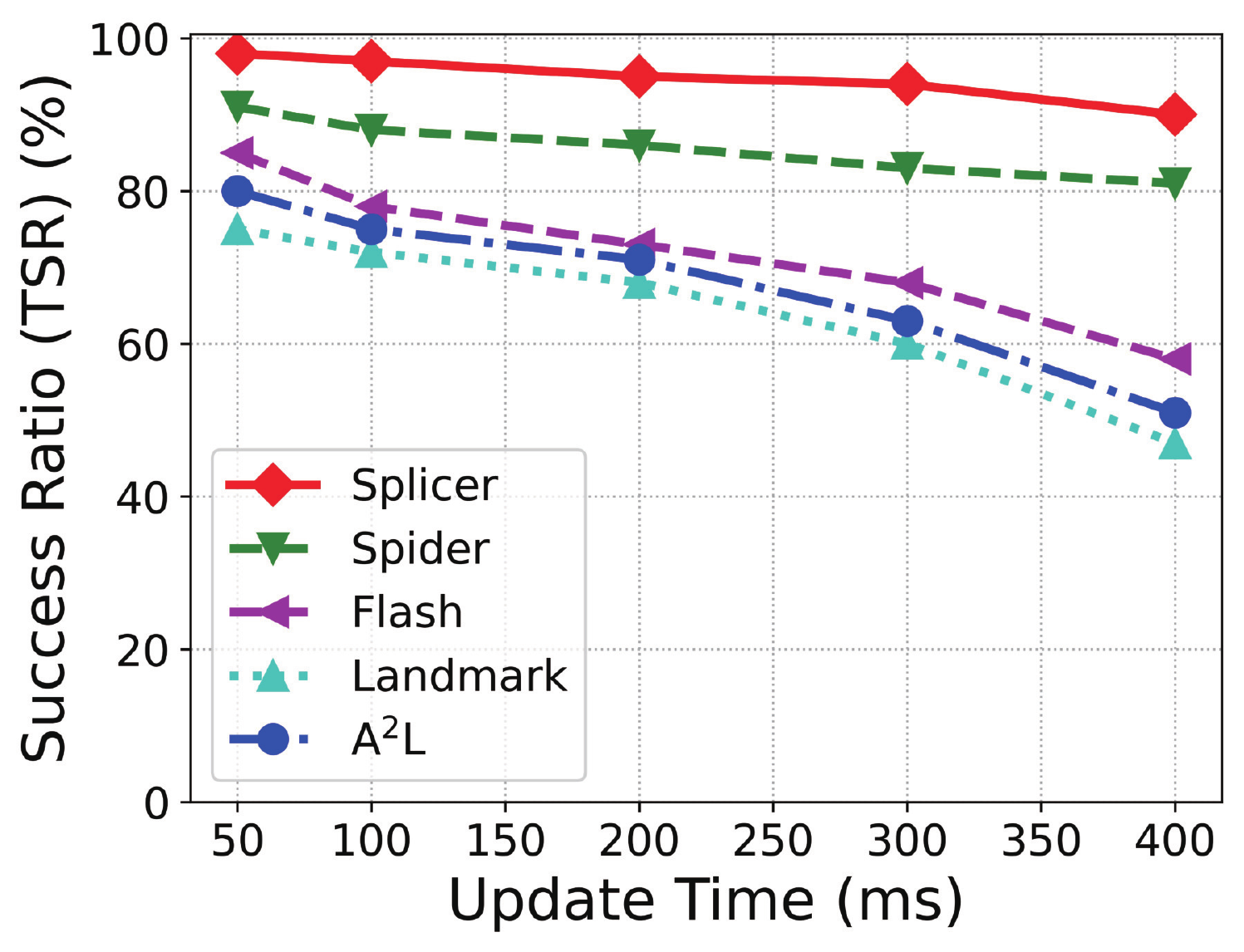}
			\label{fig_updateTime-successRatio2}
			\vspace{-1cm}
		\end{minipage}
	}%
	\hspace{1mm}
	\subfigure[Normalized throughput]{
		\begin{minipage}[t]{0.23\linewidth}
			\centering
			\includegraphics[width=1.4in]{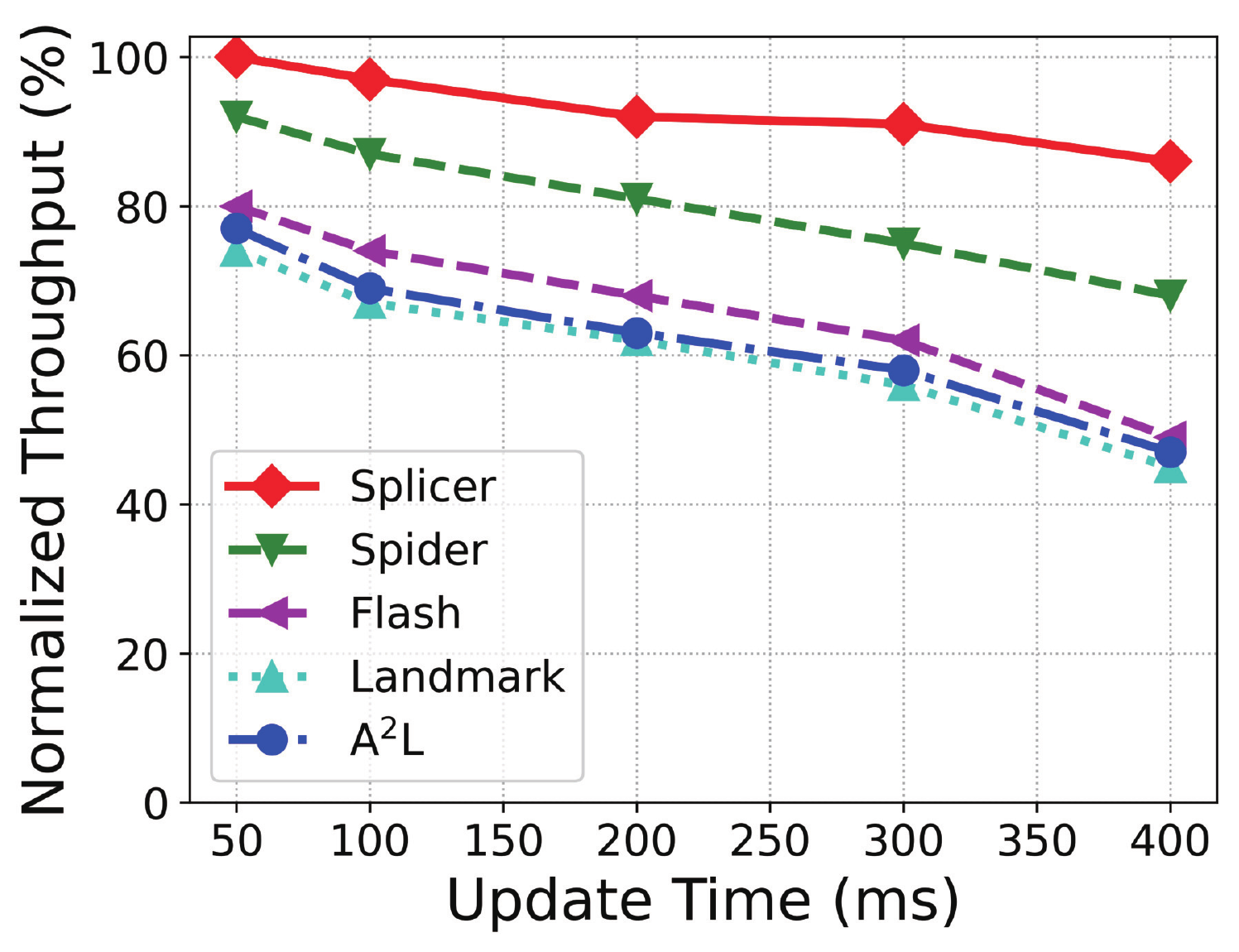}
			\label{fig_updateTime-throughput2}
			\vspace{-1cm}
		\end{minipage}
	}%
	\centering
	\vspace{-0.3cm}
	\caption{The comparison between Splicer and other schemes under different metrics in large-scale networks.}
	\label{performance2}
	\vspace{-0.5cm}
\end{figure*}

\section{Performance Evaluation}\label{Section:Evaluation}
\subsection{Experiment Setup}
Our evaluation consists of a simulation using MATLAB and full implementation of the Lightning Network Daemon (LND) testnet. We model the PCN in two scales, a small-scale network (100 nodes)  and a large-scale network (3000 nodes). Our modified LND is deployed on the machine with a six-core i7-9750H processor working at 2.6 GHz, 32 GB of RAM, 500 GB of SSD, and a 10 Gbps network interface. Referring to Spider's evaluation benchmark, the channel connections between nodes are generated by ROLL \cite{HadianNMQ16} based on the Watts-Strogatz small-world model. Following the heavy-tailed distribution of the real-world dataset on the lightning channel size \cite{TikhomirovMM20}, funds are set on each side of the channels. The directional distribution of each transaction is generated on our processed Lightning Network real-world dataset, and the transaction value is generated in the same credit card dataset \cite{Credit} adopted by Spider. Notice that we have confirmed that these transactions are guaranteed to cause some local deadlocks and contain large-value transactions that the Lightning Network cannot handle.

\textbf{Parameter settings.} The minimum, average, and median channel sizes are 10, 403, and 152 tokens. The transaction timeout is 3 seconds, the \textit{Min-TU} is 1 token, and the \textit{Max-TU} is 4 token. The number $k$ of multiple paths is 5. We set management cost $\zeta_{mn} = 0.02 \cdot hops_{mn}$, and the synchronization cost $\delta_{nl} = 0.01 \cdot hops_{nl}$, $\epsilon_{nl} = 0.05 \cdot hops_{nl}$, where $hops$ represents the number of hops in the communication path between nodes. In congestion control, we set the queue size of each channel as 8000 tokens. The factors of the window size $\beta$ and $\gamma$ are 10 and 0.1, respectively. The update time $\tau = 200$ ms. The threshold $T$ of delay in the queue is 400 ms.

\subsection{Performance of Splicer} \label{Per-Spl}
We study the performance of Splicer under different metrics compared with different schemes. As shown in Fig. \ref{performance} and Fig. \ref{performance2}, Splicer consistently outperforms other schemes in small and large network scales. \textbf{Spider} \cite{Sivaraman2020HighTC} is a multi-path source routing scheme in which each sender decides the routes. \textbf{Flash} \cite{WangXJW19} is also based on source routing, using a modified max-flow algorithm to find paths for large payments, and routing small payments randomly through precomputed paths. \textbf{Landmark} routing is adopted in many prior PCN routing schemes\cite{Flare2016, MalavoltaMKM17, RoosMKG18}. Each sender computes the shortest path to the well-connected landmark nodes, and then the landmark nodes route to the destination in $k$ distinct shortest paths. The \textbf{A$^2$L} \cite{Tairi} is the state-of-the-art PCH that focuses on providing unlinkability. The results are as follows:

\textbf{Transaction success ratio (TSR)} means the number of completed transactions over the number of generated transactions. A high value of TSR indicates the stability of the model, i.e., the ability to handle transaction deadlocks and balance the network load. Fig. \ref{fig_channelSize-successRatio} and \ref{fig_channelSize-successRatio2} show the TSR of Splicer is an average of 53.4\% higher than the other four schemes. Combining the Fig. \ref{fig_TransactionSize-successRatio} and \ref{fig_TransactionSize-successRatio2}, there is also a considerable increase (49.1\%) in the TSR as the transaction size varies. These results demonstrate that the PCHs distributed routing decision protocol can improve the TSR. We note that the improvement of Spider is more obvious under the large-scale networks. This feature benefits hubs' deployment because hubs perform many routes, have larger capital, and thus may have a larger channel size. Fig. \ref{fig_updateTime-successRatio} and \ref{fig_updateTime-successRatio2} show the TSR under the influence of update time $\tau$ in different schemes. The results show that the TSR of Splicer is stable above 90\% with the increase of update time, which is slightly higher (5\% and 10.5\%, respectively) than Spider. Because Spider also adopts a multi-path routing strategy, reducing the possibility of deadlocks, the TSR is high when the channel size is appropriate. In contrast, the TSR of A$^2$L decreased significantly, and Splicer increased by 26\% and 39\%, respectively. However, Spider handles source routing computations at the end-users, which is limited by the performance of a single machine, so the TSR is lower than Splicer, especially with large-scale networks. Because A$^2$L's complex cryptographic primitives reduce scalability, the overall average improvement in TSR for Splicer is 42\%, which proves that the network funds flow smoothly, almost without deadlock. 

\begin{figure*}[t]
	\vspace{-8mm}
	\centering
	\hspace{-5mm}
	\subfigure[Balance cost]{
		\begin{minipage}[t]{0.145\linewidth}
			\centering
			\includegraphics[width=1.15in]{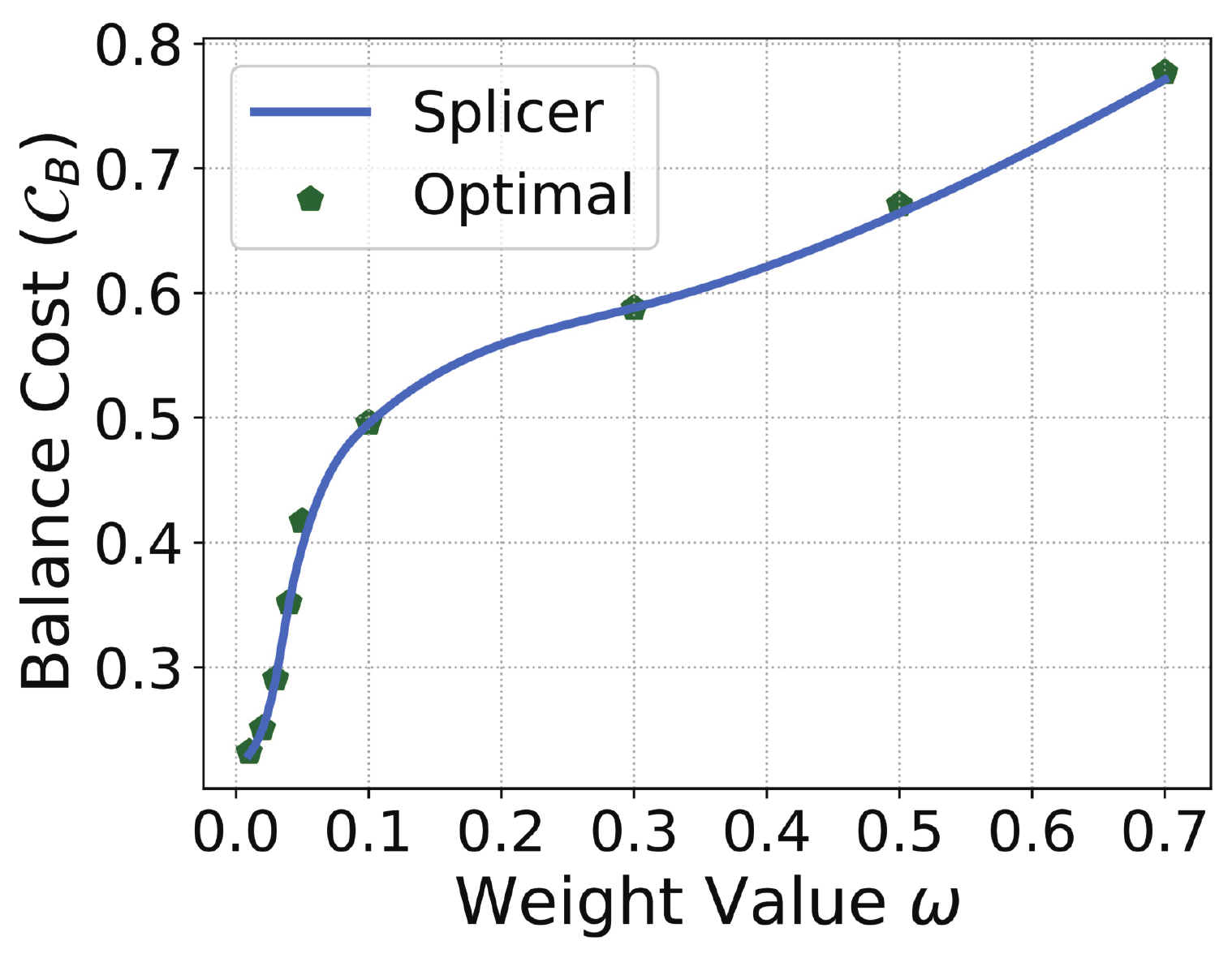}
			\label{fig_para-cost}
			\vspace{-1cm}
		\end{minipage}%
	}%
	\hspace{1mm}
	\subfigure[Trade-off in costs]{
		\begin{minipage}[t]{0.145\linewidth}
			\centering
			\includegraphics[width=1.15in]{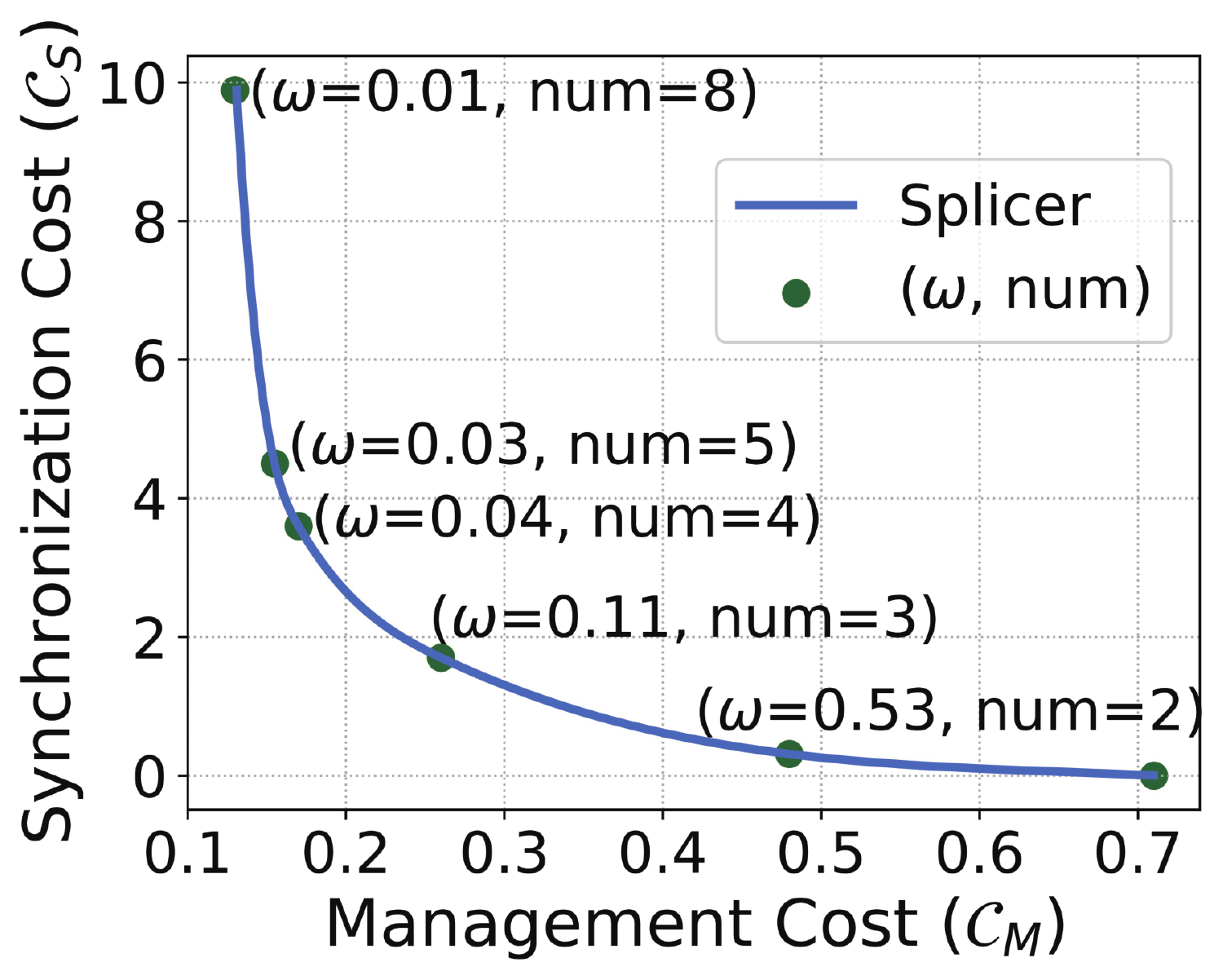}
			\label{fig_tradeoff}
			\vspace{-1cm}
		\end{minipage}%
	}%
	\hspace{1mm}
	\subfigure[Small-scale networks]{
		\begin{minipage}[t]{0.145\linewidth}
			\centering
			\includegraphics[width=1.165in]{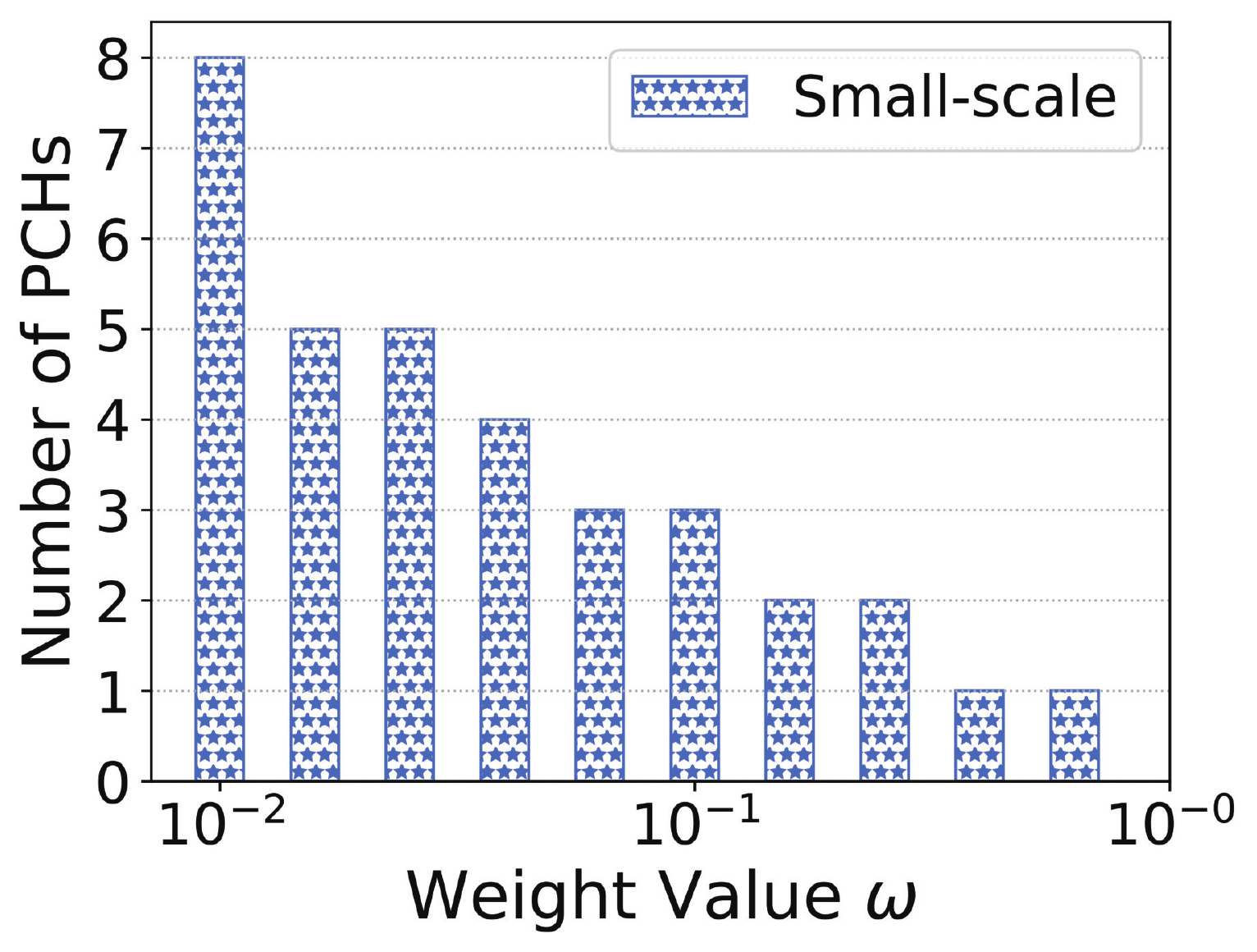}
			\label{number-small}
			\vspace{-1cm}
		\end{minipage}
	}%
	\hspace{1mm}
	\subfigure[Large-scale networks]{
		\begin{minipage}[t]{0.15\linewidth}
			\centering
			\includegraphics[width=1.24in]{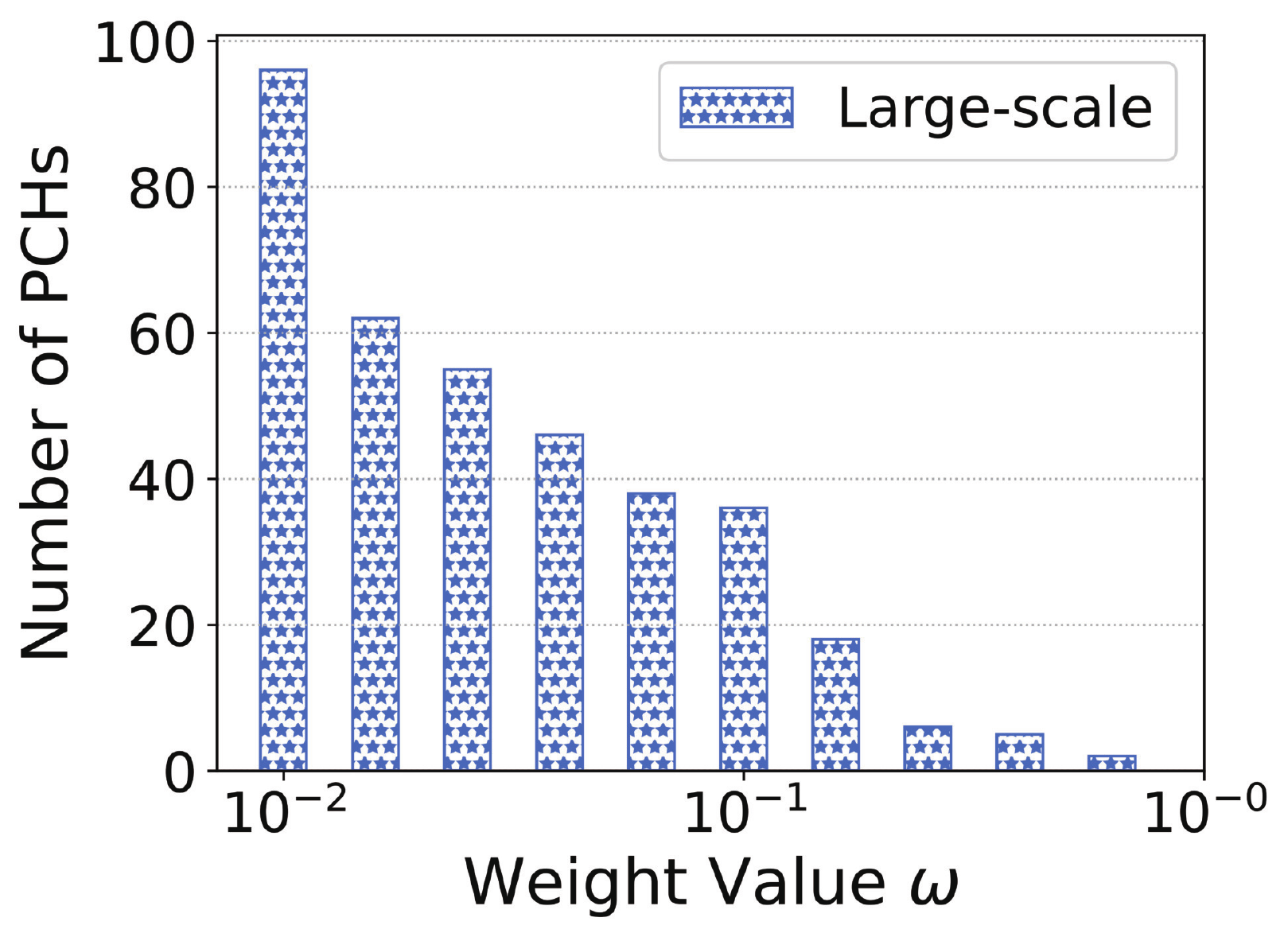}
			\label{number-large}
			\vspace{-1cm}
		\end{minipage}
	}%
	\hspace{1mm}
	\subfigure[Small-scale costs]{
		\begin{minipage}[t]{0.15\linewidth}
			\centering
			\includegraphics[width=1.24in]{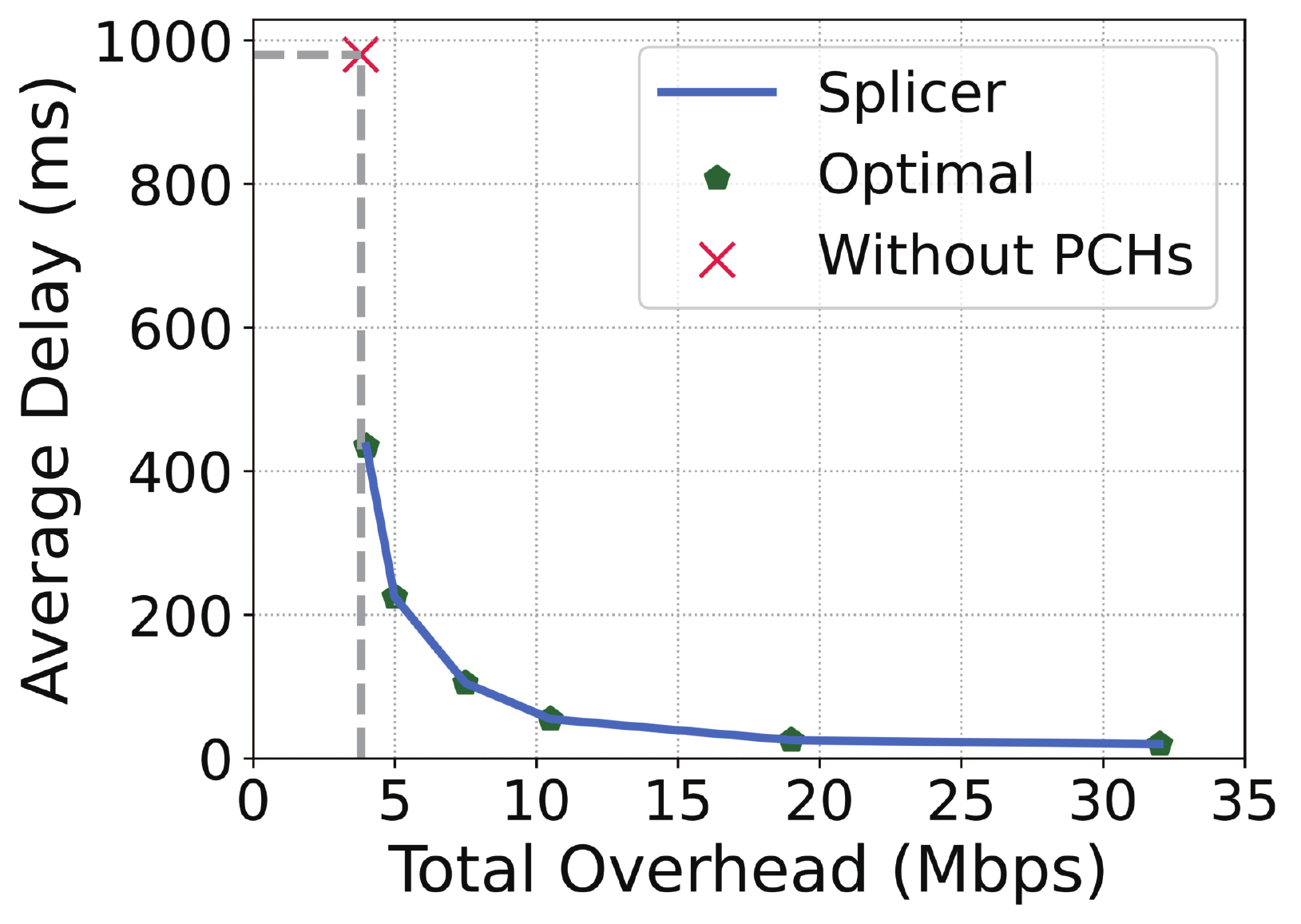}
			\label{cost-small}
			\vspace{-1cm}
		\end{minipage}
	}%
	\hspace{1mm}
	\subfigure[Large-scale costs]{
		\begin{minipage}[t]{0.15\linewidth}
			\centering
			\includegraphics[width=1.25in]{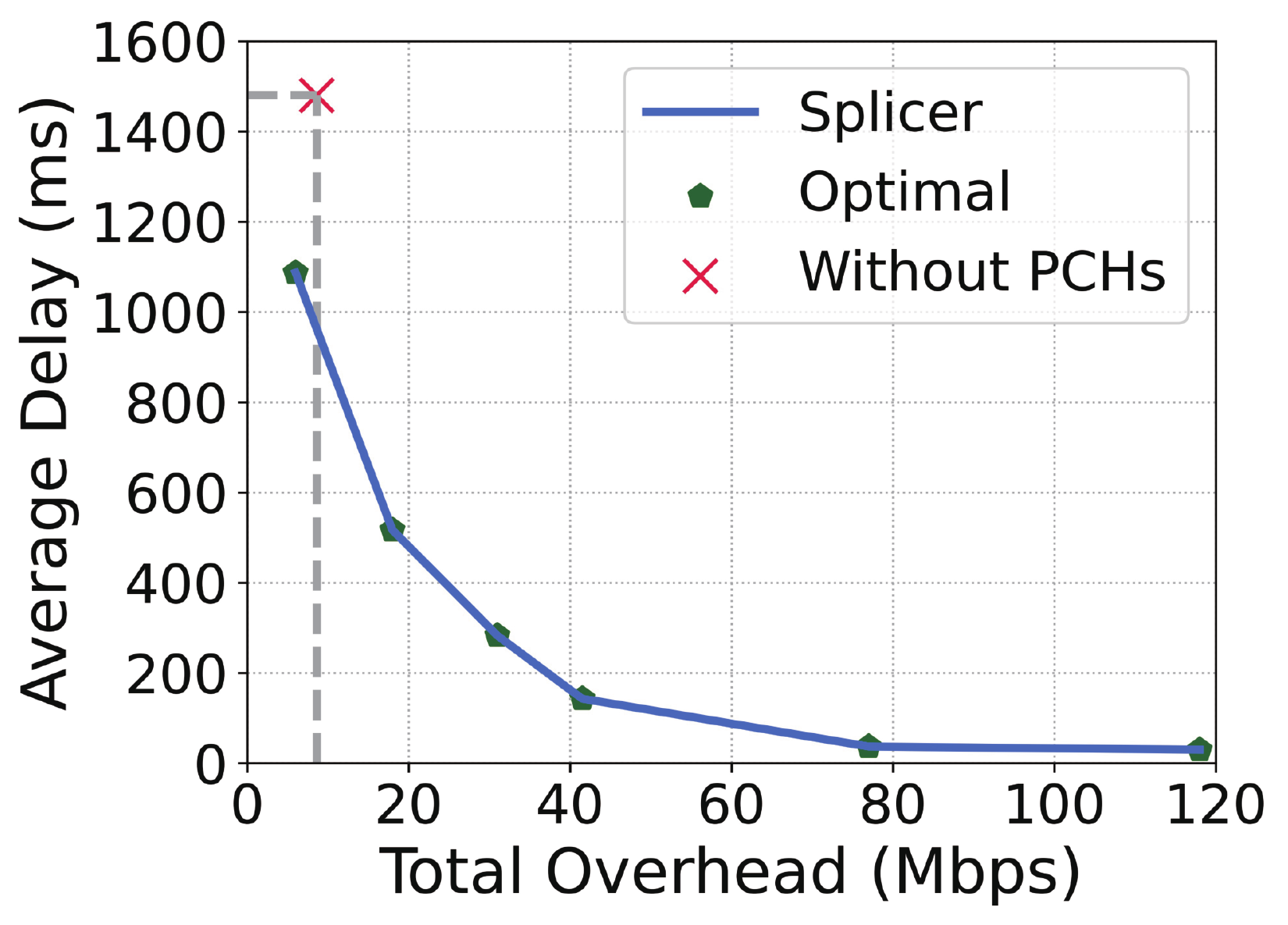}
			\label{cost-large}
			\vspace{-1cm}
		\end{minipage}
	}%
	\centering
	\vspace{-0.3cm}
	\caption{Evaluation of smooth node placement.}
	\label{placement}
	\vspace{-0.5cm}
\end{figure*}

\begin{table*}[t]
	\centering
	\caption{The influence of the different choices in routing for Splicer.}\vspace{-0.25cm}
	\resizebox{1.7\columnwidth}{!}{
		\begin{tabular}{|c|c|c|c|c|c|c|c|c|c|c|c|c|}
			\hline
			\multirow{2}[4]{*}{Scale} & \multicolumn{4}{c|}{Path Type} & \multicolumn{4}{c|}{Path Number} & \multicolumn{4}{c|}{Scheduling Algorithm} \bigstrut\\
			\cline{2-13}          & KSP   & Heuristic & EDW   & EDS   & 1     & 3     & 5     & 7     & FIFO  & LIFO  & SPF   & EDF \bigstrut\\
			\hline
			Small & 65.20\% & 77.02\% & 85.29\% & 83.26\% & 32.31\% & 70.74\% & 86.13\% & 82.88\% & 53.81\% & 90.35\% & 76.18\% & 70.44\% \bigstrut[t]\\
			Large & 58.85\% & 76.19\% & 90.07\% & 88.31\% & 35.53\% & 65.45\% & 90.40\% & 87.72\% & 61.19\% & 93.23\% & 82.21\% & 78.26\% \bigstrut[b]\\
			\hline
		\end{tabular}%
	}
	\label{tab-choices}%
	\vspace{-0.6cm}
\end{table*}%

\textbf{Normalized throughput} is the total value of payments completed over the total value generated, normalized by the maximum throughput. A high throughput demonstrates the model's ability for massive concurrent transactions and further corroborates TSR to prove the stability of the model. Fig. \ref{fig_updateTime-throughput} and \ref{fig_updateTime-throughput2} show the normalized throughput of Splicer is an average of 29.3\% higher than the other four schemes. Splicer's throughput improvement is more significant in large-scale networks (average 37.7\%). Compared to Spider, the normalized throughput of Splicer is 8.5\% and 15.6\% higher on average, respectively. Compared with A$^2$L, the improvement of Splicer is more prominent, which are 28.2\% and 48.4\%, respectively. As the increased update time, more and more transactions are getting closer to the deadline, so the probability of transaction failure is higher. Because A$^2$L lacks a scalable routing strategy design, it is more affected by this factor. Therefore, considering the TSR and throughput, we choose the median of 200 ms as the update time for Splicer.

The above results show that Splicer can significantly improve performance scalability compared to state-of-the-arts. In large-scale networks, the performance improvement effect of Splicer is more significant. Additionally, Splicer's placement optimization makes communication costs smaller (see \S \ref{Eva_SN_Pla}). Thus in large-scale low-power scenarios, we recommend Splicer.

\subsection{Evaluation of Smooth Node Placement} \label{Eva_SN_Pla}
We evaluate the placement of smooth nodes as in Fig. \ref{placement}.

\textbf{Efficiency tradeoff.} Fig. \ref{fig_para-cost} and \ref{fig_tradeoff} show the influence of weight value on the costs in the small-scale network. Fig. \ref{fig_para-cost}  shows that running the PCHs' average balance cost varies with the weight value of $\omega$ proposed in \S \ref{detail_problem}. Overall, the performance of our model is close to the optimal for almost all values of $\omega$. This indicates that our model successfully simulates the relationship between the two communication costs of the network. Fig. \ref{fig_tradeoff} further shows the tradeoff between the two costs. The annotation for the nodes in the figure is the corresponding weight $\omega$ and the \textit{number} of smooth nodes (e.g. 4 smooth nodes for $\omega = 0.04$). Management costs are incurred between smooth nodes and clients, and synchronization costs are only incurred between smooth nodes. PCNs have different affordability for these two costs. For example, because of the strong computational capability of the PCHs, PCNs can bear a high cost of synchronization. Clients may be IoT nodes so that PCNs can carry less management cost. Therefore, Splicer can adjust both costs by increasing or decreasing the number of smooth nodes in the voting smart contract suitably based on the results. In addition, the influence curves of weight value on costs in the large-scale network are similar to those in the small-scale network. The difference is that large-scale networks require more smooth nodes than small-scale networks. Fig. \ref{number-small} and \ref{number-large} show the number of smooth nodes for different weight values $\omega$ in small and large network scales. When management cost is preferred, Splicer deploys more smooth nodes, reducing the communication overhead and latency of the smooth nodes managing the clients and vice-versa.

\textbf{PCH placement effectiveness.} Fig. \ref{cost-small} and \ref{cost-large} show the average transaction delay and total traffic overhead with and without PCHs (i.e., comparing distributed routing and source routing decisions) to demonstrate the effectiveness of smooth nodes. We depict the delay-overhead curves by iterating the weight values in small and large-scale networks. If without smooth nodes, the average delay and traffic overhead are fixed. With similar total overhead, the average delay of Splicer is significantly lower than without smooth nodes. Overall, Splicer achieves 80.9\% lower latency than schemes without PCHs (e.g., Spider). Appropriate placement of some PCHs can reduce the total overhead of the network. Besides, Splicer can tolerate more traffic overhead to reduce transaction latency further.

\subsection{Routing Choices in Splicer} \label{Choices}
In addition, we study the influence of the different choices in routing on the TSR, as shown in Table \ref{tab-choices}.

\textbf{Path type.} We evaluate the performance by choosing different types of routing paths for each source-destination pair in two network scales. \textbf{KSP} means the k-shortest paths. \textbf{Heuristic} method picks 5 feasible paths with the highest channel funds. \textbf{EDW} represents the edge-disjoint widest paths. \textbf{EDS} means the edge-disjoint shortest paths. The results show that EDW outperforms other approaches in two network scales. Due to the channel size following the heavy-tailed distribution, the widest paths can better utilize the network's capacity.

\textbf{Path number.} We evaluate the performance in the different numbers of EDW paths. The TSR increases as the number of paths increases, suggesting that more paths utilize the network's capacity better. It is noted that the TSR declines slightly as the number of paths increases to 7, due to the high computational complexity that causes the performance bottleneck. Therefore, we choose 5 routing paths in Splicer.

\textbf{Scheduling algorithm.} We change the scheduling methods for the waiting queue. There are four kinds of methods: first in first out (\textbf{FIFO}); last in first out (\textbf{LIFO}); smallest payments first (\textbf{SPF}); and earliest deadline first (\textbf{EDF}). The results show that LIFO represents 10-40\% higher than other methods since it first processes the transactions far from the deadline. FIFO and EDF process transactions closest to their deadlines, resulting in poor transaction performance due to more failures. Though SPF shows the second well performance, the large transactions pile up and take up a lot of channel funds, leading to a lower transaction success ratio.

\section{Related Work}\label{section:Relatedwork}
\textbf{Payment channel hubs:} TumbleBit \cite{Ethan2017TumbleBit} presents a cryptographic protocol for the PCHs, which maintains multi-channels to reduce the routing complexity and make the transactions \textit{unlinkable} (i.e., the hubs do not know the two parties of transactions). But TumbleBit relies on scripting-based functionality, and the communication complexity increases linearly with the security parameter. A$\rm ^{2}$L \cite{Tairi} proposes a novel cryptographic primitive to improve TumbleBit, which provides better backward compatibility and efficiency. Commit-chains \cite{Rami2018CommitChains} process off-chain transactions in a pattern similar to PCHs: a centralized operator maintains a service for multiple users. It provides a tradeoff between channel establishment cost, user churn, collateral management, and decentralization. Perun \cite{Dziembowski2019Perun} proposes a smart contract-based method to build virtual PCHs to reduce communication complexity but at the cost of losing unlinkability. However, they do not consider the placement of the PCH. The PCH placement problem can significantly affect the transaction latency and communication overhead of PCNs.

\textbf{Layer-2 source routing:} Flare \cite{Flare2016} describes a hybrid routing algorithm, which seeks to optimize the average time of finding a payment route. Revive \cite{KhalilG17} proposes the first rebalancing scheme for PCNs, which rebalances the channels where node funds are unbalanced. But the routing algorithm proposed by Revive relies on a trusted third party, which is vulnerable to the single point of attack. Sprites \cite{0001BBKM19} tries to reduce the worst-case ``collateral cost" of an off-chain linked transaction. Spider \cite{Sivaraman2020HighTC} presents a multi-path routing scheme based on ``packetization", which can achieve high-throughput routing in the PCNs. However, the sender computes the routing path to the recipient. In a large-scale PCN, the performance requirements of the sender can be quite demanding. In addition, there are other new layer-2 scaling solutions, such as Rollups \cite{Rollup} proposed on Ethereum, that increase throughput by batching transactions at the expense of high latency. However, it is outside the scope of our discussion on PCN architecture. 

\section{Conclusion} \label{conclu}
We propose a distributed routing mechanism with high scalability based on multiple PCHs, to seek a new tradeoff between decentralization and scalability for PCNs. PCHs route transaction flows in PCNs in an optimal deadlock-free manner. We formulate the PCH placement problem for network scalability and propose two solutions in small and large-scale networks. To improve performance scalability, we design the rate-based routing and congestion control protocol on PCHs. Extensive experimental results show that Splicer outperforms the state-of-the-arts.

\section*{Acknowledgments}
This work was supported in part by National Key R\&D Program of China (No. 2020YFB1005500); the National Natural Science Foundation of China (No. 61972310, 61972017, 62072487, 61941114); the Beijing Natural Science Foundation (No. M21036); the Populus Euphratica Found, China (No. CCFHuaweiBC2021008); the Innovation Fund of Xidian University, China (No. YJSJ23005). 

\bibliographystyle{IEEEtran}
\normalem
\bibliography{test}

\begin{thebibliography}{10}
\providecommand{\url}[1]{#1}
\csname url@samestyle\endcsname
\providecommand{\newblock}{\relax}
\providecommand{\bibinfo}[2]{#2}
\providecommand{\BIBentrySTDinterwordspacing}{\spaceskip=0pt\relax}
\providecommand{\BIBentryALTinterwordstretchfactor}{4}
\providecommand{\BIBentryALTinterwordspacing}{\spaceskip=\fontdimen2\font plus
\BIBentryALTinterwordstretchfactor\fontdimen3\font minus
  \fontdimen4\font\relax}
\providecommand{\BIBforeignlanguage}[2]{{%
\expandafter\ifx\csname l@#1\endcsname\relax
\typeout{** WARNING: IEEEtran.bst: No hyphenation pattern has been}%
\typeout{** loaded for the language `#1'. Using the pattern for}%
\typeout{** the default language instead.}%
\else
\language=\csname l@#1\endcsname
\fi
#2}}
\providecommand{\BIBdecl}{\relax}
\BIBdecl

\bibitem{LN}
\BIBentryALTinterwordspacing
{The Bitcoin Lightning Network: Scalable Off-Chain Instant Payments}. [Online].
  Available: \url{{https://lightning.network/lightning-network-paper.pdf}}
\BIBentrySTDinterwordspacing

\bibitem{RD}
\BIBentryALTinterwordspacing
{Raiden network}. [Online]. Available: \url{{https://raiden.network}}
\BIBentrySTDinterwordspacing

\bibitem{Ethan2017TumbleBit}
E.~Heilman, L.~Alshenibr, F.~Baldimtsi, A.~Scafuro, and S.~Goldberg,
  ``Tumblebit: An untrusted bitcoin-compatible anonymous payment hub.'' in
  \emph{NDSS}, 2017.

\bibitem{Tairi}
E.~Tairi, P.~Moreno-Sanchez, and M.~Maffei, ``{A2L: Anonymous Atomic Locks for
  Scalability in Payment Channel Hubs},'' in \emph{2021 2021 IEEE Symposium on
  Security and Privacy (SP)}, 2021, pp. 1834--1851.

\bibitem{EOS}
\BIBentryALTinterwordspacing
{EOSIO Blockchain}. [Online]. Available: \url{{https://eos.io/}}
\BIBentrySTDinterwordspacing

\bibitem{Flare2016}
M.~S. A.~O. Pavel~Prihodko, Slava~Zhigulin and O.~Osuntokun, ``Flare: An
  approach to routing in lightning network,'' 2016.

\bibitem{0001BBKM19}
A.~Miller, I.~Bentov, S.~Bakshi, R.~Kumaresan, and P.~McCorry, ``Sprites and
  state channels: Payment networks that go faster than lightning.'' in
  \emph{Financial Cryptography}, vol. 11598, 2019, pp. 508--526.

\bibitem{KhalilG17}
R.~Khalil and A.~Gervais, ``Revive: Rebalancing off-blockchain payment
  networks,'' in \emph{Proceedings of the 2017 ACM SIGSAC Conference on
  Computer and Communications Security}, ser. CCS '17, 2017, p. 439–453.

\bibitem{Sivaraman2020HighTC}
V.~Sivaraman, S.~B. Venkatakrishnan, K.~Ruan, P.~Negi, L.~Yang, R.~Mittal,
  G.~Fanti, and M.~Alizadeh, ``High throughput cryptocurrency routing in
  payment channel networks,'' in \emph{NSDI}, 2020.

\bibitem{WangXJW19}
P.~Wang, H.~Xu, X.~Jin, and T.~Wang, ``Flash: efficient dynamic routing for
  offchain networks.'' in \emph{CoNEXT}.\hskip 1em plus 0.5em minus 0.4em\relax
  ACM, 2019, pp. 370--381.

\bibitem{Dziembowski2019Perun}
S.~{Dziembowski}, L.~{Eckey}, S.~{Faust}, and D.~{Malinowski}, ``Perun: Virtual
  payment hubs over cryptocurrencies,'' in \emph{2019 IEEE Symposium on
  Security and Privacy (SP)}, 2019, pp. 106--123.

\bibitem{Rami2018CommitChains}
R.~Khalil, A.~Zamyatin, G.~Felley, P.~Moreno-Sanchez, and A.~Gervais,
  ``Commit-chains: Secure, scalable off-chain payments,'' Cryptology ePrint
  Archive, Report 2018/642, 2018.

\bibitem{lnd}
\BIBentryALTinterwordspacing
{Lightning Network Daemon}. [Online]. Available:
  \url{{https://github.com/lightningnetwork/lnd}}
\BIBentrySTDinterwordspacing

\bibitem{GRJ1999}
R.~Gennaro, S.~Jarecki, H.~Krawczyk, and T.~Rabin, ``Secure distributed key
  generation for discrete-log based cryptosystems.'' in \emph{EUROCRYPT}, vol.
  1592, 1999, pp. 295--310.

\bibitem{LiMZ20}
P.~Li, T.~Miyazaki, and W.~Zhou, ``Secure balance planning of off-blockchain
  payment channel networks.'' in \emph{INFOCOM}.\hskip 1em plus 0.5em minus
  0.4em\relax IEEE, 2020, pp. 1728--1737.

\bibitem{Ge2022ShadufNP}
Z.-L. Ge, Y.~Zhang, Y.~Long, and D.~Gu, ``Shaduf: Non-cycle payment channel
  rebalancing,'' in \emph{NDSS}, 2022.

\bibitem{CelisHV18}
L.~E. Celis, L.~Huang, and N.~K. Vishnoi, ``Multiwinner voting with fairness
  constraints.'' in \emph{IJCAI}, 2018, pp. 144--151.

\bibitem{QinPIT18}
Q.~Qin, K.~Poularakis, G.~Iosifidis, and L.~Tassiulas, ``{SDN Controller
  Placement at the Edge: Optimizing Delay and Overheads.}'' in
  \emph{INFOCOM}.\hskip 1em plus 0.5em minus 0.4em\relax IEEE, 2018, pp.
  684--692.

\bibitem{Ilev01}
V.~P. Il'ev, ``An approximation guarantee of the greedy descent algorithm for
  minimizing a supermodular set function.'' \emph{Discret. Appl. Math.}, vol.
  114, no. 1-3, pp. 131--146, 2001.

\bibitem{FeldmanNS11}
M.~Feldman, J.~Naor, and R.~Schwartz, ``Nonmonotone submodular maximization via
  a structural continuous greedy algorithm - (extended abstract).'' in
  \emph{ICALP (1)}, vol. 6755, 2011, pp. 342--353.

\bibitem{BuchbinderFNS15}
N.~Buchbinder, M.~Feldman, J.~Naor, and R.~Schwartz, ``A tight linear time
  (1/2)-approximation for unconstrained submodular maximization.'' \emph{SIAM
  J. Comput.}, vol.~44, no.~5, pp. 1384--1402, 2015.

\bibitem{KellyV05}
F.~P. Kelly and T.~Voice, ``Stability of end-to-end algorithms for joint
  routing and rate control.'' \emph{Computer Communication Review}, vol.~35,
  no.~2, pp. 5--12, 2005.

\bibitem{kelly2005stability}
F.~Kelly and T.~Voice, ``Stability of end-to-end algorithms for joint routing
  and rate control,'' \emph{ACM SIGCOMM Computer Communication Review},
  vol.~35, no.~2, pp. 5--12, 2005.

\bibitem{palomar2006tutorial}
D.~P. Palomar and M.~Chiang, ``A tutorial on decomposition methods for network
  utility maximization,'' \emph{IEEE Journal on Selected Areas in
  Communications}, vol.~24, no.~8, pp. 1439--1451, 2006.

\bibitem{HaRX08}
S.~Ha, I.~Rhee, and L.~Xu, ``{CUBIC: a new TCP-friendly high-speed TCP
  variant.}'' \emph{ACM SIGOPS Oper. Syst. Rev.}, vol.~42, no.~5, pp. 64--74,
  2008.

\bibitem{HadianNMQ16}
A.~Hadian, S.~Nobari, B.~Minaei-Bidgoli, and Q.~Qu, ``{ROLL: Fast In-Memory
  Generation of Gigantic Scale-free Networks.}'' in \emph{SIGMOD Conference},
  2016, pp. 1829--1842.

\bibitem{TikhomirovMM20}
S.~Tikhomirov, P.~Moreno-Sanchez, and M.~Maffei, ``A quantitative analysis of
  security, anonymity and scalability for the lightning network.'' in
  \emph{EuroS\&P Workshops}, 2020, pp. 387--396.

\bibitem{Credit}
\BIBentryALTinterwordspacing
{Credit Card Fraud Detection}. [Online]. Available:
  \url{{https://www.kaggle.com/mlg-ulb/creditcardfraud}}
\BIBentrySTDinterwordspacing

\bibitem{MalavoltaMKM17}
G.~Malavolta, P.~Moreno-Sanchez, A.~Kate, and M.~Maffei, ``Silentwhispers:
  Enforcing security and privacy in decentralized credit networks.'' in
  \emph{NDSS}, 2017.

\bibitem{RoosMKG18}
S.~Roos, P.~Moreno-Sanchez, A.~Kate, and I.~Goldberg, ``Settling payments fast
  and private: Efficient decentralized routing for path-based transactions.''
  in \emph{NDSS}, 2018.

\bibitem{Rollup}
\BIBentryALTinterwordspacing
{Optimistic Rollups}. [Online]. Available:
  \url{{https://docs.ethhub.io/ethereum-roadmap/layer-2-scaling/optimistic_rollups/}}
\BIBentrySTDinterwordspacing

\end{thebibliography}

\end{document}